\documentclass[runningheads,a4paper]{llncs}
\usepackage{placeins}
\usepackage{booktabs}   
\usepackage{caption}
\usepackage{amsfonts}
\usepackage{xcolor}
\usepackage{listings}
\usepackage{bnf}
\usepackage{multicol}
\usepackage{parcolumns}
\usepackage{tipa}
\usepackage{bbm}
\usepackage{mathrsfs}
\usepackage{amsmath}
\usepackage{amssymb}
\usepackage{algorithm}
\usepackage{multirow}
\usepackage{upgreek}
\usepackage[noend]{algpseudocode}
\usepackage{subcaption}
\usepackage{latexsym}

\usepackage{amssymb}
\usepackage{graphicx}

\begin{document}

\mainmatter  

\title{Quantitative Analysis of Smart Contracts}

\titlerunning{Quantitative Analysis of Smart Contracts}

%
%
\author{Krishnendu Chatterjee$^1$, Amir Kafshdar Goharshady$^1$, Yaron Velner$^2$}
\authorrunning{K. Chatterjee, A.K. Goharshady, Y. Velner}

\institute{$^1$IST Austria (Institute of Science and Technology Austria)\\ $^2$Hebrew University of Jerusalem\\
	\texttt{\{krishnendu.chatterjee, amir.goharshady\}@ist.ac.at} \\
	\texttt{yaron.welner@mail.huji.ac.il}
}

\maketitle

\begin{abstract}
Smart contracts are computer programs that are
executed by a network of mutually distrusting agents, without the
need of an external trusted authority.
Smart contracts handle and transfer assets of considerable value (in the form of crypto-currency like Bitcoin).
Hence, it is crucial that their implementation is bug-free.
We identify the utility (or expected payoff) of interacting with such smart 
contracts as the basic and canonical quantitative property for such contracts.
We present a framework for such quantitative analysis of smart contracts.
Such a formal framework poses new and novel research challenges in programming languages,
as it requires modeling of game-theoretic aspects to analyze incentives 
for deviation from honest behavior and modeling utilities which are not 
specified as standard temporal properties such as safety and termination. 
While game-theoretic incentives have been analyzed in the security community, 
their analysis has been restricted to the very special case of stateless games.
However, to analyze smart contracts, stateful analysis is required as it must account 
for the  different program states of the protocol.
Our main contributions are as follows: we present
(i)~a simplified programming language for smart contracts; 
(ii)~an automatic translation of the programs to state-based games;
(iii)~an abstraction-refinement approach to solve such games; and
(iv)~experimental results on real-world-inspired smart contracts.

\end{abstract}

\definecolor{ndkeyword}{rgb}{0.2, 0.6, 0.8}

\lstdefinelanguage{solidity}{
	keywords={returns, contract, typeof, new, true, false, catch, function, return, null, catch, switch, var, if, in, not, while, do, else, case, break, false, true, for, public, private},
	keywordstyle=\color{blue}\bfseries,
	ndkeywords={boolean, int, uint, uint256, int256, address, class, export, boolean, throw, implements, import, this, bytes, bool},
	ndkeywordstyle=\color{ndkeyword}\bfseries,
	identifierstyle=\color{black},
	sensitive=false,
	comment=[l]{//},
	morecomment=[s]{/*}{*/},
	morestring=[b]',
	morestring=[b]",
	backgroundcolor=\color{white},
	numbers=left,
	showstringspaces=false,
	showspaces=false,
	tabsize=2,
	basicstyle=\scriptsize\ttfamily,
	keywordstyle=\color{blue}\bfseries,
	xleftmargin=2.0ex,
	numberstyle=\scriptsize,numbersep=1pt,
	commentstyle=\color{brownbright},
	escapeinside={/*@}{@*/}
}

\lstdefinelanguage{ours}{
	keywords={null, function, payout, payable, caller, if, else, and, or, return},
	keywordstyle=\color{blue}\bfseries,
	ndkeywords={numeric, map, id},
	ndkeywordstyle=\color{ndkeyword}\bfseries,
	identifierstyle=\color{black},
	sensitive=false,
	comment=[l]{//},
	morecomment=[s]{/*}{*/},
	morestring=[b]',
	morestring=[b]",
	backgroundcolor=\color{white},
	showstringspaces=false,
	showspaces=false,
	tabsize=2,
	basicstyle=\scriptsize\ttfamily,
	keywordstyle=\color{blue}\bfseries,
	xleftmargin=2.0ex,
	numberstyle=\scriptsize,numbersep=1pt,
	commentstyle=\color{red},
	escapeinside={/*@}{@*/}
}

\newcommand{\Xup}{\overline{X}}
\newcommand{\E}{\mathbb{E}}
\newcommand{\Xdown}{\underline{X}}
\newcommand{\Rup}{\overline{R}}
\newcommand{\Rdown}{\underline{R}}
\newcommand{\xup}{\overline{x}}
\newcommand{\xdown}{\underline{x}}
\newcommand{\Tup}{\overline{T}}
\newcommand{\Tdown}{\underline{T}}
\newcommand{\tup}{\overline{t}}
\newcommand{\tdown}{\underline{t}}
\newcommand{\Xdefault}{\mathring{X}}
\newcommand{\xdefault}{\mathring{x}}
\newcommand{\pho}{\rho}
\newcommand{\akhar}{\blacksquare}
\newcommand{\avval}{\square}
\newcommand{\Eta}{H}
\newcommand{\hich}{\textsc{Null}}
\newcommand{\function}{{\color{blue}\textbf{\texttt{function }}}}
\newcommand{\kw}[1]{{\color{blue}\textbf{\texttt{#1}}}}
\newcommand{\moves}{\mathcal{M}}
\newcommand{\naught}{\boxtimes}
\newcommand{\Supp}{\textit{{Supp}}}
\newcommand{\cond}{\textit{cond}}
\newcommand{\ENd}{\textit{end}}
\newcommand{\Movecall}{\textit{call}}
\newcommand{\Movepay}{\textit{pay}}
\newcommand{\Movedecide}{\textit{decide}}
\newcommand{\rah}{\mathscr{P}}
\newcommand{\Histories}{\mathscr{H}}
\newcommand{\purestrategy}{\upvarphi}
\newcommand{\strategy}{\sigma}
\newcommand{\gamevalue}{\upupsilon}
\newcommand{\matgame}{\mathcal{G}}
\newcommand{\parties}{\mathbb{P}}
\newcommand{\party}{\mathbbmss{p}}
\newcommand{\last}{\mathit{last}}
\newcommand{\Prob}{\textup{\textsf{Prob}}}
\newcommand{\fhlen}{\mathsf{L}}
\makeatletter
\def\BState{\State\hskip-\ALG@thistlm}
\makeatother
\newcommand{\val}{\textit{val}}
\newcommand{\up}{\uparrow}
\newcommand{\down}{\downarrow}
\newcommand{\ab}{\texttt{a}}
\newcommand{\refines}{\sqsubseteq}
\newcommand{\partition}{\Uppi}
\newcommand{\dummy}{\textsl{d}}
\newcommand{\Dummy}{\textsl{D}}
\newcommand{\objects}{\mathcal{O}}
\newcommand{\balance}{\text{\textbeta}}
\renewcommand{\a}{\text{a}}
\renewcommand{\b}{\text{b}}
\vspace{-2em}
\section{Introduction}
\vspace{-0.5em}
In this work we present a quantitative stateful game-theoretic framework for formal analysis of smart-contracts.

\smallskip\noindent{\em Smart contracts.}
Hundreds of crypto-currencies are in use today, and investments in them are increasing steadily~\cite{coinmarketcap}.
These currencies are not controlled by any central authority like governments or banks, instead they are governed by the \emph{blockchain} protocol, which dictates the rules and determines the outcomes, 
e.g., the validity of money transactions and account balances.
Blockchain was initially used for peer-to-peer Bitcoin payments~\cite{nakamoto2008bitcoin}, but recently it is also used for running programs (called smart contracts).
A \emph{smart contract} is a program that runs on the blockchain, which enforces its correct execution (i.e., that it is running as originally programmed).
This is done by encoding semantics in crypto-currency transactions.
For example, Bitcoin transaction scripts allow users to specify conditions, or contracts, 
which the transactions must satisfy prior to acceptance.
Transaction scripts can encode many useful functions, such as validating that a payer owns a coin she is spending or enforcing rules for multi-party transactions.
The Ethereum crypto-currency~\cite{buterin2013ethereum} allows arbitrary stateful Turing-complete conditions over the transactions which gives rise to smart contracts that can implement a wide range of applications, such as financial instruments (e.g., financial derivatives or wills) or autonomous governance applications (e.g., voting systems).
The protocols are globally specified and their implementation is decentralized. Therefore, there is no central authority and they are immutable.
Hence, the economic consequences of critical bugs in a smart contract cannot be reverted.

\smallskip\noindent{\em  Types of Bugs.} 
There are two types of bugs with monetary consequences:
\begin{enumerate}
\item {\em Coding errors.}
 Similar to standard programs, bugs could arise from coding mistakes.
At one reported case~\cite{ethercamp},
mistakenly replacing $\texttt{+=}$ operation with $\texttt{=+}$ enabled loss of tokens that were backed by \$800,000 of investment.

\item {\em Dishonest interaction incentives.}
Smart contracts do not fully dictate the behavior of participants.
 They only specify the outcome (e.g., penalty or rewards) of the behaviors.
Hence, a second source for bugs is the high level {\em interaction aspects} that could give a participant unfair advantage and incentive for dishonest behavior.
For example, a naive design of rock-paper-scissors game~\cite{delmolino2015step} allows playing sequentially, rather then concurrently, and gives advantage to the second player who can see the opponent's move.
\end{enumerate}

\smallskip\noindent{\em DAO attack: interaction of two types of bugs.}
Quite interestingly a coding bug can incentivize dishonest behavior as in the famous DAO attack~\cite{MIT_REVIEW}.
The Decentralized Autonomous Organization (DAO)~\cite{DAO_white_paper} is an Ethereum smart contract~\cite{DAO_problem}.
The contract consists of investor-directed venture capital fund.
On June 17, 2016 an attacker exploited a bug in the contract to extract \$80 million~\cite{MIT_REVIEW}.
Intuitively, the root cause was that the contract allowed users to first get hold of their funds, and only then updated their balance records while a semantic detail allowed the attacker to withdraw multiple times before the update.

\smallskip\noindent{\em Necessity of formal framework.}
Since bugs in smart contracts have direct economic consequences and are irreversible, 
they have the same status as safety-critical errors for programs and reactive systems and
must be detected before deployment.
Moreover, smart contracts are deployed rapidly. 
 There are over a million smart contracts in Ethereum, holding over 15 billion dollars at the time of writing~\cite{Etherscan}.
It is impossible for security researchers to analyze all of them, and 
 lack of automated tools for programmers makes them error prone.
Hence, a formal analysis framework for smart contract bugs is of great importance.

\smallskip\noindent{\em Utility analysis.} 
In verification of programs, specifying objectives is non-trivial and a key goal is to consider specification-less verification, where basic
properties are considered canonical.
For example, termination is a basic property in program analysis; and 
data-race freedom or serializability are basic properties in concurrency. 
Given these properties, models are verified wrt them without considering
any other specification.
For smart contracts, describing the correct specification that prevents dishonest behavior is more challenging due to the presence of game-like interactions. 
We propose to consider the expected user utility (or payoff) 
that is guaranteed even in presence of adversarial behavior of other agents as a canonical property.
Considering malicious adversaries is standard in 
game theory.
For example, the expected utility of a fair lottery is $0$.
An analysis reporting a different utility signifies a bug.

\smallskip\noindent{\em New research challenges.}
Coding bugs are detected by classic verification, program analysis, and model checking tools~\cite{ClarkeBook,SMC08}.
However, a formal framework for incentivization bugs 
presents a new research challenge for the programming language community.
Their analysis must overcome two obstacles:
(a)~the framework will have to handle game-theoretic aspects to model interactions and incentives for dishonest behavior; and 
(b)~it will have to handle properties that cannot be deduced from standard temporal properties such as safety or termination, 
but require analysis of monetary gains (i.e., quantitative properties).

While game-theoretic incentives are widely analyzed by the security community~(e.g., see~\cite{bonneau2015sok}), their analysis 
is typically restricted to the very special case of one-shot games that do not consider different states of the program,
and thus the consequences of decisions on the next state of the program 
are ignored.
In addition their analysis is typically ad-hoc and stems from brainstorming and special techniques.
This could work when very few protocols existed (e.g., when bitcoin first emerged) and deep thought was put into making them elegant and analyzable.
However, the fast deployment of smart contracts makes it crucial to automate the process and make it accessible to programmers.

\smallskip\noindent{\em Our contribution.}
In this work we present a formal framework for quantitative analysis of utilities in smart contracts. Our contributions are as follows:
\begin{enumerate}
\item We present a simplified (loop-free) programming language that allows game-theoretic interactions.
We show that many classical smart contracts can be easily described in our language, and conversely,
a smart contract programmed in our language can be easily translated to Solidity~\cite{solidity}, which is 
the most popular Ethereum smart contract language. 

\item The underlying mathematical model for our language is stateful concurrent games.
We automatically translate programs in our language to such games.

\item The key challenge to analyze such game models automatically is to tackle the state-space 
explosion.
While several abstraction techniques have been considered for programs~\cite{queille1982specification,godefroid1996partial,burch1992symbolic}, 
they do not work for game-theoretic models with quantitative objectives.
We present an approach based on interval-abstraction for reducing the states, establish soundness of our abstraction, and present a refinement process. This is our core technical contribution.

\item We present experimental results on several classic real-world smart contracts.
We show that our approach can handle contracts that otherwise give rise to 
games with up to $10^{23}$ states. While special cases of concurrent games (namely, turn-based games) have been studied in verification and reactive synthesis, there are no practical methods to solve general concurrent quantitative games. To the the best of our knowledge, there are no tools to solve quantitative concurrent games other than academic examples of few states, and we present the first practical method to solve quantitative concurrent games that scales to real-world smart contract analysis.

\end{enumerate}
In summary, our contributions range from (i)~modeling of smart contracts as state-based games, 
to (ii)~an abstraction-refinement approach to solve such games, to (iii)~experimental results 
on real-world smart contracts.

\smallskip\noindent{\em Organization.}
We start with an overview of smart contracts in Section~\ref{sec:smart_contracts}.
Our programming language is introduced in Section~\ref{sec:prog_lang} along with implementations of real-world contracts in this language.
We then present state-based concurrent games and 
translation of contracts to games in Section~\ref{sec:games}.
The abstraction-refinement methodology for games is presented in Section~\ref{sec:abstract} followed by experimental results in 
Section~\ref{sec:exp}.
Section~\ref{sec:comparison_new} presents a comparison with related work and Section~\ref{sec:conclusion} concludes the paper with suggestions for future research.

\newcommand{\programVar}[1]{\mathit{#1}}
\vspace{-1em}
\section{Background on Ethereum smart contracts} \label{sec:smart_contracts}
\vspace{-0.5em}
\subsection{Programmable smart contracts}
\vspace{-0.5em}
Ethereum~\cite{buterin2013ethereum} is a decentralized virtual machine, which runs programs called
contracts.
Contracts are written in a Turing-complete
bytecode language, called Ethereum Virtual Machine (EVM) bytecode~\cite{wood2014ethereum}.
A contract is invoked by calling one of its functions, where each function is defined by a sequence of instructions.
The contract maintains a persistent internal state and can receive (transfer) currency from (to) users and other contracts.
Users send transactions to the Ethereum network to invoke functions.
Each transaction may contain input parameters for the contract and an associated monetary amount, possibly $0$, which is transferred from the user to the contract.

Upon receiving a transaction, the contract collects the money sent to it, executes a function according to input parameters, and updates its internal state.
All transactions are recorded on a decentralized ledger, called blockchain. A sequence of transactions that begins from the creation of the network uniquely determines the state of each contract and balances of users and contracts.
The blockchain does not rely on a trusted central authority, rather, each transaction is processed by a large network
of mutually untrusted peers called miners.
Users constantly broadcast transactions to the network.
Miners add transactions to the blockchain via a proof-of-work consensus protocol~\cite{nakamoto2008bitcoin}.

We illustrate contracts by an example (Figure~\ref{fig:satcontract}) which implements a contract that rewards users who solve a satisfiability problem.
Rather than programming it directly as EVM bytecode, we use Solidity, a widely-used programming language which compiles into EVM bytecode~\cite{solidity}.
This contract has one variable $\programVar{balance}$.
Users can send money by calling $\programVar{deposit}$, and the $\programVar{balance}$ is updated according to the sent amount, which is specified by $\programVar{msg.value}$ keyword.
Users submit a solution by calling $\programVar{submitSolution}$ and giving values to input parameters $\programVar{a},\programVar{b},\programVar{c}$ and $\programVar{d}$.
If the input is a valid solution, the user who called the function, denoted by the keyword $\programVar{msg.sender}$, is rewarded by $\programVar{balance}$.
We note that when several users submit valid solutions, only the \emph{first} user will be paid.
Formally, if two users $A$ and $B$ submit transactions $t_A$ and $t_B$ to the contract that invoke $\programVar{submitSolution}$ and have valid input parameters, then user $A$ is paid if and only if $t_A$ appears in the blockchain before $t_B$.

\begin{figure}
	\vspace{-8mm}
\begin{lstlisting}[language=solidity]
contract SAT {
    uint balance;
    function deposit() payable {
        balance += msg.value;
    }  
    function submitSolution( bool a, bool b, bool c, bool d ) {
        if( (a || b || !c) && (!a || c || !d) {
            msg.sender.send(balance);
            balance = 0;
        }
    }}
\end{lstlisting}
\vspace{-5mm}
\caption{Smart contract that rewards users satisfying 
$(a \vee b \vee \neg c) \wedge (\neg a \vee c \vee \neg d)$.
}
\label{fig:satcontract}
\vspace{-5mm}
\end{figure}

\smallskip\noindent{\em Subtleties.}
In this work, for simplicity, we ignore some details in the underlying protocol of Ethereum smart contract.
We briefly describe these details below:
\begin{itemize}
\item \emph{Transaction fees.} In exchange for including her transactions in the blockchain, a user pays transaction fees to the miners, proportionally to the execution time of her transaction.
This fact could slightly affect the monetary analysis of the user gain, but could also introduce bugs in a program, as there is a bound on execution time that cannot be exceeded. Hence, it is possible that some functions could never be called, or even worse, a user could actively give input parameters that would prevent other users from invoking a certain function.

\item \emph{Recursive invocation of contracts.}
A contract function could invoke a function in another contract, which in turn can have a call to the original contract.
The underling Ethereum semantic in recursive invocation was the root cause for the notorious DAO hack~\cite{DAO_analysis}.
 
\item \emph{Behavior of the miners.}
Previous works have suggested that smart contracts could be implemented to encourage miners to deviate from their honest behavior~\cite{teutschcryptocurrencies}. 
This could in theory introduce bugs into a contract, e.g., a contract might give unfair advantage for a user who is a big miner.

\end{itemize}
\vspace{-1em}
\subsection{Tokens and user utility}
\vspace{-0.5em}
A user's utility is determined by the Ether she spends and receives, but could also be affected by the state of the contract.
Most notably, smart contracts are used to issue \emph{tokens}, which can be viewed as a stake in a company or an organization, in return to an Ether (or tokens) investment (see an example in Figure~\ref{fig:token}).
These tokens are \emph{transferable} among users and are traded in exchanges in return to Ether, Bitcoin and Fiat money.
At the time of writing, smart contracts instantiate tokens worth billions of dollars~\cite{EtherscanTokens}.
Hence, gaining or losing tokens has clear utility for the user.
At a larger scope, user utility could also be affected by more abstract storage changes. Some users would be willing to pay to have a contract declare them as Kings of Ether~\cite{king}, while others could gain from registering their domain name in a smart contract storage~\cite{ens}.
In the examples provided in this work we mainly focus on utility that arises from Ether, tokens and the like.
However, our approach is general and can model any form of utility by introducing auxiliary utility variables and definitions.

\begin{figure}
	\vspace{-7mm}
\begin{lstlisting}[language=solidity]
contract Token {
    mapping(address=>uint) balances;
    function buy() payable {
        balances[msg.sender] += msg.value;
    }
    function transfer( address to, uint amount ) {
        if(balances[msg.sender]>=amount) {
            balances[msg.sender] -= amount;
            balances[to] += amount;
    }}}
\end{lstlisting}
\vspace{-4mm}
\caption{Token contract example.}
\label{fig:token}
\vspace{-7mm}
\end{figure}

\vspace{-1em}
\section{Programming Language for Smart Contracts} \label{sec:prog_lang} \label{SEC:PROG_LANG}
\vspace{-0.5em}
In this section we present our programming language for smart contracts that supports
concurrent interactions between parties.
A party denotes an agent that decides to interact with the contract.
A contract is a tuple $C = (N, I, M, R, X_0, F, T)$ where $X := N \cup I \cup M$ is a set of variables, 
$R$ describes  the range of values that can be stored in each variable, $X_0$ is the initial values stored in variables, 
$F$ is a list of functions and $T$ describes for each function, the time segment in which it can be invoked. 
We now formalize these concepts.

\smallskip\noindent{\em Variables.} There are three distinct and disjoint types of variables in $X$:

\begin{itemize}
	\item $N$ contains ``numeric'' variables that can store a single integer.

	\item $I$ contains ``identification'' (``id'') variables capable of pointing to a party in the contract by her address or storing $\hich$. 
	The notion of ids is quite flexible in our approach: The only dependence on ids is that they should be distinct and an id should not act on behalf of another id. 
	We simply use different integers to denote distinct ids and assume that a ``faking of identity'' does not happen. In Ethereum this is achieved by digital signatures. 

	\item $M$ is the set of ``mapping'' variables. Each $m \in M$ maps parties to integers.
\end{itemize}

\smallskip\noindent{\em Bounds and Initial values.} The tuple $R = (\Rdown, \Rup)$ where $\Rdown, \Rup: N \cup M \rightarrow \mathbb{Z}$ represent lower and upper bounds for integer values that can be stored in a variable. 
For example, if $n \in N$, then $n$ can only store integers between $\Rdown(n)$ and $\Rup(n)$. Similarly, if $m \in M$ is a mapping and $i \in I$ stores an address to a party in the contract, then $m\left[i\right]$ can save integers between $\Rdown(m)$ and $\Rup(m)$. The function $X_0 : X \rightarrow \mathbb{Z} \cup \{\hich\}$ assigns an initial value to every variable. The assigned value is an integer in case of numeric and mapping variables, i.e., a mapping variable maps everything to its initial value by default. Id variables can either be initialized by $\hich$ or an id used by one of the parties.

\smallskip\noindent{\em Functions and Timing.} The sequence $F = <f_1, f_2, \ldots, f_n>$ is a list of functions and $T = (\Tdown, \Tup)$, where $\Tdown, \Tup: F \rightarrow \mathbb{N}$. The function $f_i$ can only be invoked in time-frame $T(f_i) = \left[ \Tdown(f_i), \Tup(f_i) \right]$. The contract uses a global clock, for example the current block number in the blockchain, to keep track of time. 

Note that we consider a single contract, and interaction between multiple contracts is a subject of future work.
\vspace{-1em}
\subsection{Syntax}\label{sec:syntax}
\vspace{-0.5em}
We provide a simple overview of our contract programming language. Our language is syntactically similar to Solidity and a translation mechanism for different aspects is discussed in Section~\ref{sec:translation}. An example contract, modeling a game of rock-paper-scissors, is given in Figure \ref{prog:rps}. Here, a party, called \texttt{issuer} has issued the contract and taken the role of \texttt{Alice}. Any other party can join the contract by registering as \texttt{Bob} and then playing rock-paper-scissors. To demonstrate our language, we use a bidding mechanism. A more exact treatment of the syntax using a formal grammar can be found in Appendix~\ref{app:syntax}. 

\begin{figure}[!tbh]
	\noindent\begin{minipage}{.45\textwidth}
		\begin{lstlisting}[language=ours]
		(0) contract RPS {
		map Bids[0, 100] = 0;
		id Alice = issuer;
		id Bob = null;
		numeric played[0,1] = 0;
		numeric AliceWon[0,1] = 0;
		numeric BobWon[0,1] = 0;
		numeric bid[0, 100] = 0;
		numeric AlicesMove[0,3] = 0; 
		numeric BobsMove[0,3] = 0; 
		//0 denotes no choice,
		//1 rock, 2 paper, 
		//3 scissors
				
		(1) function registerBob[1,10]
			   (payable bid : caller) {
		(2)	   if(Bob==null) {
		(3)		   Bob = caller;
		(4)	     Bids[Bob]=bid;
		       }
			     else{
		(5)	     payout(caller, bid);
			     }
		(6) }		
		(7) function play[11, 15]
			  (AlicesMove:Alice = 0,
			  BobsMove:Bob = 0,
			  payable Bids[Alice]: Alice){   
		(8)  if(played==1)
		(9)	    return;
			   else
		(10)	  played = 1;
	\end{lstlisting}
	\end{minipage}
	\noindent\begin{minipage}{.45\textwidth}
		\begin{lstlisting}[language=ours]
		(11) if(BobsMove==0 and AlicesMove!=0)
		(12)	    AliceWon = 1;
		(13) else if(AlicesMove==0 and BobsMove!=0)
		(14)	    BobWon = 1;
		(15) else if(AlicesMove==0 and BobsMove==0)
			    {
		(16)	    AliceWon = 0;
		(17)	    BobWon = 0;
			    }
		(18) else if(AlicesMove==BobsMove+1 or
				    AlicesMove==BobsMove-2)
		(19)	    AliceWon = 1;
			    else
		(20)	    BobWon = 1;
		(21) }	
		
		(22) function getReward[16,20]() {
		(23)	if(caller==Alice and AliceWon==1
			    or caller==Bob and BobWon==1)
					{
		(24)		payout(caller, Bids[Alice] + Bids[Bob]);
		(25)		Bids[Alice] = 0;
		(26)		Bids[Bob] = 0;
					}
	 	(27) }
		    }
		\end{lstlisting}
	\end{minipage}
	\vspace{-3mm}
	\caption{A rock-paper-scissors contract.}
	\label{prog:rps}
	\vspace{-5mm}
\end{figure}

\smallskip\noindent{\em Declaration of Variables.} The program begins by declaring variables\footnote{For simplicity, we demonstrate our method with global variables only. However, the method is applicable to general variables as long as their ranges are well-defined at each point of the program.}, their type, name, range and initial value. For example, \texttt{Bids} is a map variable that assigns a value between $0$ and $100$ to every id. This value is initially $0$. Line numbers (labels) are defined in Section \ref{sec:semantics} below and are not part of the syntax.

\smallskip\noindent{\em Declaration of Functions.} After the variables, the functions are defined one-by-one. Each function begins with the keyword \function followed by its name and the time interval in which it can be called by parties. Then comes a list of input parameters. Each parameter is of the form \texttt{ variable : party } which means that the designated party can choose a value for that variable. The chosen value is required to be in the range specified for that variable. The keyword \kw{caller}
denotes the party that has invoked this function and \kw{payable} signifies that the party should not only decide a value, but must also pay the amount she decides. For example, \texttt{registerBob} can be called in any time between $1$ and $10$ by any of the parties. At each such invocation the party that has called this function must pay some amount which will be saved in the variable \texttt{bid}. After the decisions and payments are done, the contract proceeds with executing the function.

\smallskip\noindent{\em Types of Functions.} There are essentially two types of functions, depending on their parameters. \emph{One-party functions}, such as \texttt{registerBob} and \texttt{getReward} require parameters from \kw{caller} only, while \emph{multi-party functions}, such as \texttt{play} ask several, potentially different, parties for input. In this case all parties provide their input decisions and payments concurrently and without being aware of the choices made by other parties, also a default value is specified for every decision in case a relevant party does not take part. 

\smallskip\noindent{\em Summary.} Putting everything together, in the contract specified in Figure \ref{prog:rps}, any party can claim the role of Bob between time $1$ and time $10$ by paying a bid to the contract, if the role is not already occupied. Then at time $11$ one of the parties calls \texttt{play} and both parties have until time $15$ to decide which choice (rock, paper, scissors or none) they want to make. Then the winner can call \texttt{getReward} and collect her prize.

\vspace{-1em}
\subsection{Semantics} \label{sec:semantics}
\vspace{-0.5em}
In this section we present the details of the semantics. 
In our programming language there are several key aspects which are 
non-standard in programming languages, such as the notion of time progress,
concurrency, and interactions of several parties. 
Hence we present a detailed description of the semantics.
We start with the requirements.

\smallskip\noindent\emph{Requirements.} 
In order for a contract to be considered valid, other than following the syntax rules, a few more requirements must be met, which are as follows:

\begin{itemize}
\item We assume that no division by zero or similar undefined behavior happens.
\item To have a well-defined message passing, we also assume that no multi-party function has an associated time interval intersecting that of another function. 
\item Finally, for each non-id variable $v$, it must hold that $\Rdown(v) \leq X_0(v) \leq \Rup(v)$ and similarly, for every function $f_i$, we must have $\Tdown(f_i) < \Tup(f_i)$.
\end{itemize} 

\smallskip\noindent\emph{Overview of time progress.} Initially, the time is $0$. Let $F_t$ be the set of functions executable at time $t$, i.e., $F_t = \{f_i \in F \vert t \in T(f_i) \}$, then 
$F_t$ is either empty or contains one or more one-party functions or consists of a single multi-party function. We consider the following cases:

\begin{itemize}
\item {\em $F_t$ empty.}
If $F_t$ is empty, then nothing can happen until the clock ticks.

\item \emph{Execution of one-party functions.} If $F_t$ contains one or more one-party functions, then each of the parties can call any subset of these functions at time $t$. If there are several calls at the same time, the contract might run them in any order.
  While a function call is being executed, all parties are able to see the full state of the contract, and can issue new calls. When there are no more requests for function calls, the clock ticks and the time is increased to $t+1$. When a call is being executed and is at the beginning part of the function, its caller can send messages or payments to the contract. Values of these messages and payments will then be saved in designated variables and the execution continues. If the caller fails to make a payment or specify a value for a decision variable or if her specified values/payments are not in the range of their corresponding variables, i.e.~they are too small or too big, the call gets canceled and the contract reverts any changes to variables due to the call and continues as if this call had never happened. 

\item\emph{Execution of multi-party functions.} If $F_t$ contains a single multi-party function $f_i$ and $t < \Tup(f_i)$, then any party can send messages and payments to the contract to specify values for variables that are designated to be paid or decided by her. These choices are hidden and cannot be observed by other participants. She can also change her decisions as many times as she sees fit. The clock ticks when there are no more valid requests for setting a value for a variable or making a payment. This continues until we reach time $\Tup(f_i)$. At this time parties can no longer change their choices and the choices become visible to everyone. The contract proceeds with execution of the function. If a party fails to make a payment/decision or if $\hich$ is asked to make a payment or a decision, default behavior will be enforced. Default value for payments is $0$ and default behavior for other variables is defined as part of the syntax. For example, in function \texttt{play} of Figure \ref{prog:rps}, if a party does not choose, a default value of $0$ is enforced and given the rest of this function, this will lead to a definite loss. 
\end{itemize} 

Given the notion of time progress we proceed to formalize the notion of ``runs'' of the contract. 
This requires the notion of labels, control-flow graphs, valuations, and states, which we describe below.

\smallskip\noindent\emph{Labels.} Starting from $0$, we give the contract, beginning and end points of every function, and every command a label. The labels are given in order of appearance. As an example, see the labels in parentheses in Figure \ref{prog:rps}.

\smallskip\noindent\emph{Entry and Exit Labels.} We denote the first (beginning point) label in a function $f_i$ by $\avval_i$ and its last (end point) label by $\akhar_i$. 

\smallskip\noindent \emph{Control Flow Graphs (CFGs).} We define the control flow graph $CFG_i$ of the function $f_i$ in the standard manner, i.e.~$CFG_i = (V, E)$, where there is a vertex corresponding to every labeled entity inside $f_i$. We do not distinguish an entity, its label and its corresponding vertex. Each edge $e \in E$ has a condition $\cond(e)$ which is a boolean expression that must be true when traversing that edge. For example, Figure \ref{fig:cfgplay} is an illustration of the control flow graph of function \texttt{play} in our example contract. For a more formal treatment see Appendix \ref{app:semantics}.

\begin{figure}[h]
	\vspace{-4mm}
	\begin{minipage}{0.30\textwidth}
			\includegraphics[scale=0.03]{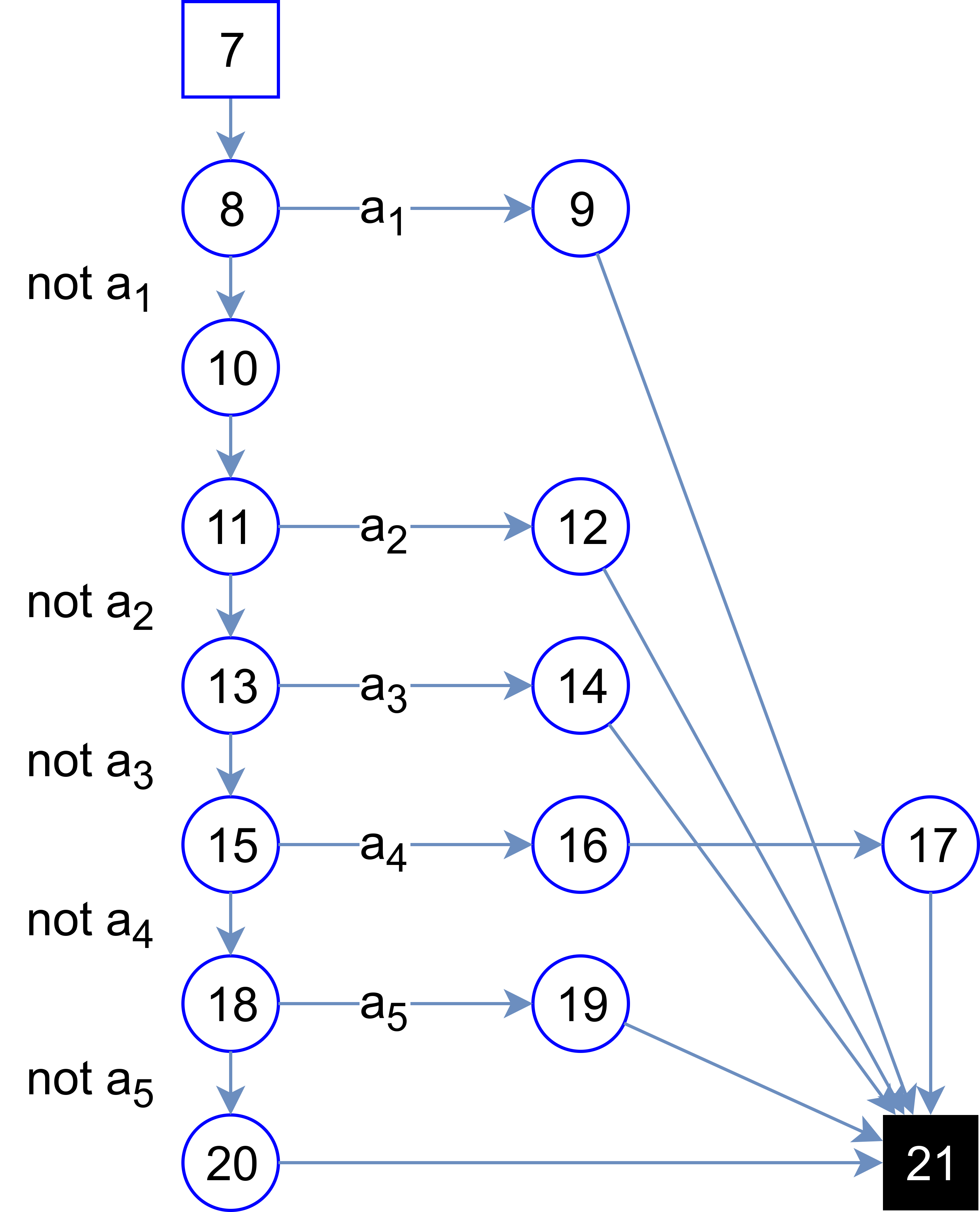}
	\end{minipage}
	\begin{minipage}{0.40\textwidth}
		\begin{tabular}{ll}
		$\text{a}_1 : $ & \texttt{played==1}\\
		$\text{a}_2 : $ & \texttt{BobsMove==0 and AlicesMove!=0}\\
		$\text{a}_3 : $ & \texttt{AlicesMove==0 and BobsMove!=0}\\
		$\text{a}_4 : $ & \texttt{AlicesMove==0 and BobsMove==0}\\
		$\text{a}_5 : $ & \texttt{AlicesMove==BobsMove+1} \\
		 & \texttt{or} \\ 
		 & \texttt{AlicesMove==BobsMove-2}\\
		\end{tabular}
	\end{minipage}
	\vspace{-2mm}
	\caption{Control Flow Graph of \texttt{play()} (left) and its edge conditions (right)}
	\label{fig:cfgplay}
	\vspace{-5mm} 
\end{figure}

\smallskip\noindent \emph{Valuations.} A valuation is a function $\val$, assigning a value to every variable. Values for numeric variables must be integers in their range, values for identity variables can be party ids or $\hich$ and a value assigned to a map variable $m$ must be a function $\val(m)$ such that for each identity $i$, we have $\Rdown(m) \leq \val(m)(i) \leq \Rup(m)$. Given a valuation, we extend it to expressions containing mathematical operations in the straight-forward manner.

\smallskip\noindent \emph{States.} A state of the contract is a tuple $s = (t, b, l, \val, c)$, where $t$ is a time stamp, $b \in \mathbb{N} \cup \{0\}$ is the current balance of the contract, i.e., the total amount of payment to the contract minus the total amount of payouts, $l$ is a label (that is being executed), $\val$ assigns values to variables and $c \in P \cup \{ \perp \}$, is the caller of the current function. $c = \perp$ corresponds to the case where the caller is undefined, e.g., when no function is being executed. We use $S$ to denote the set of all states that can appear in a run of the contract as defined below.

\smallskip\noindent \emph{Runs.} A run $\rho$ of the contract is a finite sequence $\left\{\rho_j = (t_j, b_j, l_j, \val_j, c_j) \right\}_{j=0}^{r}$ of states, starting from $(0, 0, 0, X_0, \perp)$, that follows all rules of the contract and ends in a state with time-stamp $t_{r} > \max_{f_i} \Tup(f_i)$. These rules must be followed when switching to a new state in a run: 
\begin{itemize}
	\item The clock can only tick when there are no valid pending requests for running a one-party function or deciding or paying in multi-party functions.
	\item Transitions that happen when the contract is executing a function must follow its control flow graph and update the valuation correctly.
	\item No variable can contain an out-of-bounds value. If an overflow or underflow happens, the closest possible value will be saved. This rule also ensures that the contract will not create new money, given that paying more than the current balance of the contract results in an underflow.
	\item Each party can call any set of the functions at any time.
\end{itemize} 
This definition is formalized in Appendix \ref{app:run}.

\begin{remark}
	Note that in our semantics each function body completes its execution in a single tick of the clock. However, ticks might contain more than one function call and execution.
\end{remark}

\smallskip\noindent \emph{Run prefixes.} We use $\Eta$ to mean the set of all prefixes of runs and denote the last state in $\eta \in \Eta$ by $\ENd(\eta)$. A run prefix $\eta'$ is an extension of $\eta$ if it can be obtained by adding one state to the end of $\eta$.

\smallskip\noindent\emph{Probability Distributions.} Given a finite set $\mathcal{X}$, a probability distribution on $\mathcal{X}$ is a function $\delta: \mathcal{X} \rightarrow \left[ 0,1 \right]$ such that $\sum_{x \in \mathcal{X}} \delta(x) = 1$. Given such a distribution, its support, $\Supp(\delta)$, is the set of all $x \in \mathcal{X}$ such that $\delta(x) > 0$. We denote the set of all probability distributions on $\mathcal{X}$ by $\Delta(\mathcal{X})$. 

Typically for programs it suffices to define runs for the semantics. However, given that there are several parties in
contracts, their semantics depends on the possible choices of the parties. 
Hence we need to define policies for parties, and such policies will define probability distribution over runs,
which constitute the semantics for contracts.
To define policies we first define moves.

\smallskip\noindent \emph{Moves.} We use $\moves$ for the set of all moves. The moves that can be taken by parties in a contract can be summarized as follows: 
\begin{itemize}
	\item Calling a function $f_i$, we denote this by $\Movecall(f_i)$.
	\item Making a payment whose amount, $y$ is saved in $x$, we denote this by $\Movepay(x, y)$.
	\item Deciding the value of $x$ to be $y$, we denote this by $\Movedecide(x, y)$.
	\item Doing none of the above, we denote this by $\naught$.
\end{itemize} 

\smallskip\noindent \emph{Permitted Moves.} We define $P_i : S \rightarrow \moves$, so that $P_i(s)$ is the set of permitted moves for the party with identity $i$ if the contract is in state $s = (t, b, l, \val, p_j)$. It is formally defined as follows:
\begin{itemize}
	\item If $f_k$ is a function that can be called at state $s$, then $\Movecall(f_k) \in P_i(s)$.
	\item If $l = \avval_q$ is the first label of a function $f_q$ and $x$ is a variable that can be decided by $i$ at the beginning of the function $f_q$, then $\Movedecide(x, y) \in P_i(s)$ for all permissible values of $y$. Similarly if $x$ can be paid by $i$, $\Movepay(x,y) \in P_i(s)$.
	\item $\naught \in P_i(s)$.
\end{itemize}

\smallskip\noindent \emph{Policies and Randomized Policies.} A policy $\pi_i$ for party $i$ is a function $\pi_i : \Eta \rightarrow A$, such that for every $\eta \in \Eta$, $\pi_i(\eta) \in P_i(\ENd(\eta))$. 
Intuitively, a policy is a way of deciding what move to use next, given the current run prefix.
A policy profile $\pi = (\pi_i)$ is a sequence assigning one policy to each party $i$.
The policy profile $\pi$ defines a unique run $\pho^\pi$ of the contract which is obtained when parties choose their moves according to $\pi$.
A randomized policy $\xi_i$ for party $i$ is a function $\xi_i : \Eta \rightarrow \Delta(\moves)$, such that $\Supp(\xi_i(s)) \subseteq P_i(s)$. 
A randomized policy assigns a probability distribution over all possible moves for party $i$ given the current run prefix of the contract, then the party can follow it by choosing a move randomly according to the distribution.
We use $\Xi$ to denote the set of all randomized policy profiles, $\Xi_i$ for randomized policies of $i$ and $\Xi_{-i}$ to denote the set of randomized policy profiles for all parties except $i$.
A randomized policy profile $\xi$ is a sequence $(\xi_i)$ assigning one randomized policy to each party.
Each such randomized policy profile induces a unique probability measure on the set of runs, which is denoted as $\Prob^{\xi}\left[\cdot\right]$. 
We denote the expectation measure associated to $\Prob^{\xi}\left[\cdot\right]$  by $\E^{\xi}\left[\cdot\right]$.

\vspace{-1em}
\subsection{Objective function and values of contracts}

As mentioned in the introduction we identify expected payoff as the canonical property for contracts.
The previous section defines expectation measure given randomized policies as the basic semantics.
Given the expected payoff, we define values of contracts as the worst-case guaranteed payoff for 
a given party.
We formalize the notion of objective function (the payoff function).

\smallskip\noindent \emph{Objective Function.} An objective $o$ for a party $p$ is in one of the following forms: 
\begin{itemize}
	\item  $(p^+ - p^-)$, where $p^+$ is the total money received by party $p$ from the contract (by ``payout'' statements) and $p^-$ is the total money paid by $p$ to the contract (as ``payable'' parameters). 
	\item An expression containing mathematical and logical operations (addition, multiplication, subtraction, integer division, and, or, not) and variables chosen from the set $N \cup \{m\left[i\right] \vert m \in M, i \in I \}$. Here $N$ is the set of numeric variables, $m[i]$'s are the values that can be saved inside maps.\footnote{We are also assuming, as in many programming languages, that $\textsc{True} = 1$ and $\textsc{False} = 0$.}  
	\item A sum of the previous two cases.
\end{itemize} 
Informally, $p$ is trying to choose her moves so as to maximize $o$.

\smallskip\noindent \emph{Run Outcomes.} Given a run $\pho$ of the program and an objective $o$ for party $p$, the outcome $\kappa(\pho, o, p)$ is the value of $o$ computed using the valuation at $\ENd(\pho)$ for all variables and accounting for payments in $\pho$ to compute $p^+$ and $p^-$.

\smallskip\noindent \emph{Contract Values.} 
Since we consider worst-case guaranteed payoff, we consider that there 
is an objective $o$ for a single party $p$ which she tries to maximize and all other parties are adversaries who aim to minimize $o$. 
Formally, given a contract $C$ and an objective $o$ for party $p$, we define the value of contract as:
$$
\mathsf{V}(C, o, p) := \sup_{\xi_p \in \Xi_p} \inf_{\xi_{-p} \in \Xi_{-p}} \E^{(\xi_p, \xi_{-p})} \left[\kappa(\pho, o, p) \right],
$$
This corresponds to $p$ trying to maximize the expected value of $o$ and all other parties maliciously colluding to minimize it.
In other words, it provides the worst-case guarantee for party~$p$, irrespective of the behavior of the other parties,
which in the worst-case is adversarial to party~$p$.

\vspace{-1em}
\subsection{Examples}\label{sec:examples}
\vspace{-0.5em}

One contribution of our work is to present the simplified programming language, and to show 
that this simple language can express several classical smart contracts.
To demonstrate the applicability, we present several examples of classical smart contracts in 
this section. 
In each example, we present a contract and  a ``buggy'' implementation of the same contract that has a different value. 
In Section~\ref{sec:exp} we show that our automated approach to analyze the contracts can compute contract values with 
enough precision to differentiate between the correct and the buggy implementation.
All of our examples are motivated from well-known bugs that have happened in real life in Ethereum.
\vspace{-1.5em}
\subsubsection{Rock-Paper-Scissors.}
Let our contract be the one specified in Figure \ref{prog:rps} and assume that we want to analyze it from the point of view of the issuer $p$. Also, let the objective function be 
$
\left( p^+ - p^- + 10 \cdot \texttt{AliceWon} \right).
$
Intuitively, this means that winning the rock-paper-scissors game is considered to have an additional value of $10$, other than the spending and earnings. The idea behind this is similar to the case with chess tournaments, in which players not only win a prize, but can also use their wins to achieve better ``ratings", so winning has extra utility.

A common bug in writing rock-paper-scissors is allowing the parties to move sequentially, rather than concurrently \cite{delmolino2015step}. If parties can move sequentially and the issuer moves after \texttt{Bob}, then she can ensure a utility of $10$, i.e.~her worst-case expected reward is $10$. However, in the correct implementation as in Figure~\ref{prog:rps}, the best strategy for both players is to bid $0$ and then Alice can win the game with probability $1/3$ by choosing each of the three options with equal probability. Hence, her worst-case expected reward is $10/3$.
\vspace{-1em}
\subsubsection{Auction.}
\vspace{-0.5em}
Consider an open auction, in which during a fixed time interval everyone is allowed to bid for the good being sold and everyone can see others' bids. When the bidding period ends a winner emerges and every other participant can get their money back. Let the variable \texttt{HighestBid} store the value of the highest bid made at the auction. Then for a party $p$, one can define the objective as:
$$
p^+ - p^- + (\texttt{Winner==}p)\times \texttt{HighestBid}.
$$
This is of course assuming that the good being sold is worth precisely as much as the highest bid. A correctly written auction should return a value of $0$ to every participant, because those who lose the auction must get their money back and the party that wins pays precisely the highest bid. The contract in Figure \ref{prog:auction} (left) is an implementation of such an auction. However, it has a slight problem. The function bid allows the winner to reduce her bid. This bug is fixed in the contract on the right.
\begin{figure}[h]
\vspace{-6mm}
	\noindent\begin{minipage}{.45\textwidth}
		\begin{lstlisting}[language=ours]
		contract BuggyAuction {
		map Bids[0,1000] = 0;
		numeric HighestBid[0,1000] = 0;
		id Winner = null; 
		numeric bid[0,1000] = 0;
		
		function bid[1,10]
		(payable bid : caller) {
		   payout(caller, Bids[caller]);
		   Bids[caller]=bid;
		   if(bid>HighestBid)
		   {
			    HighestBid = bid;
			    Winner = caller;
		   }
		}
		
		function withdraw[11,20]()
		{
			if(caller!=Winner)
			{
				payout(caller, Bids[caller]);
				Bids[caller]=0;
			}
		}}		
		\end{lstlisting}
	\end{minipage}
	\noindent\begin{minipage}{.45\textwidth}
		\begin{lstlisting}[language=ours]
		contract Auction {
		map Bids[0,1000] = 0;
		numeric HighestBid[0,1000] = 0;
		id Winner = null; 
		numeric bid[0,1000] = 0;
		
		function bid[1,10]
		(payable bid : caller) {
			if(bid<Bids[caller])
				return;
			payout(caller, Bids[caller]);
			Bids[caller]=bid;
			if(bid>HighestBid)
			{
				HighestBid = bid;
				Winner = caller;
			}
		}
		
		function withdraw[11,20]()
		{
			if(caller!=Winner)
			{
			payout(caller, Bids[caller]);
			Bids[caller]=0;
			}
		}}
		\end{lstlisting}
	\end{minipage}
	\vspace{-5mm}
	\caption{A buggy auction contract (left) and its fixed version (right).}
	\label{prog:auction}
	\vspace{-5mm}
\end{figure} 

\vspace{-1.5em}
\subsubsection{Three-Way Lottery.} Consider a three-party lottery contract issued by $p$ as in Figure \ref{prog:lottery} (left). Note that division is considered to be integer division, discarding the remainder. The other two players can sign up by buying tickets worth $1$ unit each. Then each of the players is supposed to randomly and uniformly choose a nonce. A combination of these nonces produces the winner with equal probability for all three parties. If a person does not make a choice or pay the fees, she will certainly lose the lottery. The rules are such that if the other two parties choose the same nonce, which is supposed to happen with probability $\frac{1}{3}$, then the issuer wins. Otherwise the winner is chosen according to the parity of sum of nonces. This gives everyone a winning probability of $\frac{1}{3}$ if all sides play uniformly at random. However, even if one of the sides refuses to play uniformly at random, the resulting probabilities of winning stays the same because each side's probability of winning is independent of her own choice assuming that others are playing randomly.

\begin{figure}[!htbp]
	\vspace{-1cm}
	\noindent\begin{minipage}{.55\textwidth}

		\begin{lstlisting}[language=ours]
		contract BuggyLottery {
		id issuer = p;
		id Alice = null;
		id Bob = null;
		id Winner = null;
		numeric deposit[0,1] = 0;
		numeric AlicesChoice[0,3] = 0;
		numeric BobsChoice[0,3] = 0;
		numeric IssuersChoice[0,3] = 0;
		numeric sum[0,9]=0;
		
		function buyTicket[1,10]
		  (payable deposit:caller)
		  {
		    if(deposit!=1) return;
		    if(Alice==null) Alice=caller;
		    else if(Bob==null) Bob=caller;
		    else payout(caller, deposit);	
		  }
		
		function play[11,20]
		(AlicesChoice:Alice = 0,
		BobsChoice:Bob = 0,
		IssuersChoice:issuer = 0,
		payable deposit:issuer)
		{
		if(AlicesChoice==0 or BobsChoice==0)
		  Winner = issuer;
		else if(IssuersChoice==0 or
		deposit==0)
		  Winner = Alice;
		else if(AlicesChoice==BobsChoice)
		  Winner = issuer;
		else
		{
		  sum= AlicesChoice+BobsChoice 
		  + IssuersChoice;
		  if(sum/2*2==sum)
		    Winner = Alice;
		  else
		    Winner = Bob;
		}
		}
		
		function withdraw[21,30]()
		{
		  if(caller==Winner)
		    payout(caller, 3);
		}}		
		\end{lstlisting}
	\end{minipage}
	\noindent\begin{minipage}{.45\textwidth}
		\begin{lstlisting}[language=ours]
		contract Lottery {
		id issuer = p;
		id Alice = null;
		id Bob = null;
		id Winner = null;
		numeric deposit[0,1] = 0;
		numeric AlicesChoice[0,3] = 0;
		numeric BobsChoice[0,3] = 0;
		numeric IssuersChoice[0,3] = 0;
		numeric sum[0,9]=0;
		
		function buyTicket[1,10]
		(payable deposit:caller)
		{
		  if(deposit!=1) return;
		  if(Alice==null) Alice=caller;
		  else if(Bob==null) Bob=caller;
		  else payout(caller, deposit);	
		}
		
		function play[11,20]
		(AlicesChoice:Alice = 0,
		BobsChoice:Bob = 0,
		IssuersChoice:issuer = 0,
		payable deposit:issuer)
		{
		if(AlicesChoice==0 or BobsChoice==0)
		  Winner = issuer;
		else if(IssuersChoice==0 or
		deposit==0)
		  Winner = Alice;
		else
		{
		  sum = AlicesChoice + BobsChoice 
		  + IssuersChoice;
		  if(sum/3*3==sum)
		    Winner = Alice;
		  else if(sum/3*3==sum-1)
		    Winner = Bob;
		  else
		    Winner = issuer;
		}	
		}
		
		function withdraw[21,30]()
		{
		  if(caller==Winner)
		    payout(caller, 3);
		}}		
		\end{lstlisting}
	\end{minipage}
	\vspace{-3mm}
	\caption{A buggy lottery contract (left) and its fixed version (right).}
	\label{prog:lottery}
\end{figure}

We assume that the issuer $p$ has objective $p^+ - p^-$. This is because the winner can take other players' money. In a bug-free contract we will expect the value of this objective to be $0$, given that winning has a probability of $\frac{1}{3}$. However, the bug here is due to the fact that other parties can collude. For example, the same person might register as both \texttt{Alice} and \texttt{Bob} and then opt for different nonces. This will ensure that the issuer loses. The bug can be solved as in the contract in Figure \ref{prog:lottery} (right). In that contract, one's probability of winning is $\frac{1}{3}$ if she honestly plays uniformly at random, no matter what other parties do.
 
\vspace{-1.5em}
\subsubsection{Token Sale.} Consider a contract that sells \emph{tokens} modeling some aspect of the real world, e.g.~shares in a company. At first anyone can buy tokens at a fixed price of $1$ unit per token. However, there are a limited number of tokens available and at most $1000$ of them are meant to be sold. The tokens can then be transferred between parties, which is the subject of our next example. For now, Figure \ref{prog:sale} (left) is an implementation of the selling phase. However, there is a big problem here. The problem is that one can buy any number of tokens as long as there is at least one token remaining. For example, one might first buy $999$ tokens and then buy another $1000$. If we analyze the contract from the point of view of a solo party $p$ with objective $\texttt{balance}[p]$, then it must be capped by $1000$ in a bug-free contract, while the process described above leads to a value of $1999$. The fixed contract is in Figure \ref{prog:sale} (right). This bug is inspired by a very similar real-world bug described in \cite{yaronaudit}.

\vspace{-5mm}
\subsubsection{Token Transfer.} Consider the same bug-free token sale as in the previous example, we now add a function for transferring tokens. An owner can choose a recipient and an amount less than or equal to her balance and transfer that many tokens to the recipient. Figure \ref{prog:tran} (left) is an implementation of this concept. Taking the same approach and objective as above, we expect a similar result. However, there is again an important bug in this code. What happens if a party transfers tokens to herself? She gets free extra tokens! This has been fixed in the contract on the right. This example models a real-world bug as in \cite{tokentransferbug}.
\begin{figure}[h]
	\vspace{-7mm}
	\noindent\begin{minipage}{.45\textwidth}
		\begin{lstlisting}[language=ours]
		contract BuggySale {
		map balance[0,2000] = 0;
		numeric remaining[0,2000] = 1000;
		numeric payment[0,2000] = 0;
		
		function buy[1,10]
		  (payable payment:caller)
		{
		  if(remaining<=0){
		    payout(caller, payment);
		    return;
		  }
		  balance[caller] += payment;
		  remaining -= payment;
		}}		
		\end{lstlisting}
	\end{minipage}
	\noindent\begin{minipage}{.45\textwidth}
		\begin{lstlisting}[language=ours]
		contract Sale {
		map balance[0,2000] = 0;
		numeric remaining[0,2000] = 1000;
		numeric payment[0,2000] = 0;
		
		function buy[1,10]
		  (payable payment:caller)
		{
		  if(remaining-payment<0){
		    payout(caller, payment);
		    return;
		  }
		  balance[caller] += payment;
		  remaining -= payment;
		}}
		\end{lstlisting}
	\end{minipage}
	\vspace{-5mm}
	\caption{A buggy token sale (left) and its fixed version (right).}
	\label{prog:sale}
	\vspace{-5mm}
\end{figure} 

\begin{figure}[h]
	\noindent\begin{minipage}{.55\textwidth}
		\begin{lstlisting}[language=ours]
		contract BuggyTransfer {
		map balance[0,2000] = 0;
		numeric remaining[0,2000] = 1000;
		numeric payment[0,2000] = 0;
		numeric amount[0,2000] = 0;
		numeric fromBalance[0,2000] = 0;
		numeric toBalance[0,2000] = 0;
		id recipient = null;
		
		function buy[1,10]...
		
		function transfer[1,10](
			recipient : caller
			amount : caller) {
			  fromBalance = balance[caller];
			  toBalance = balance[recipient];
			  if(fromBalance<amount)
				  return;
			  fromBalance -= amount;
			  toBalance += amount;
			  balance[caller] = fromBalance;
			  balance[recipient] = toBalance;
		  }}
		\end{lstlisting}
	\end{minipage}
	\noindent\begin{minipage}{.45\textwidth}
		\begin{lstlisting}[language=ours]
		contract Transfer {
		map balance[0,2000] = 0;
		numeric remaining[0,2000] = 1000;
		numeric payment[0,2000] = 0;
		numeric amount[0,2000] = 0;

		
		id recipient = null;
		
		function buy[1,10]...
		
		function transfer[1,10](
			recipient : caller
			amount : caller) {
			 
			  
			  if(balance[caller]<amount)
				  return;
			  balance[caller] -= amount;
			  balance[recipient] += amount;
			  
			  
		  }}
		\end{lstlisting}
	\end{minipage}
	\vspace{-4mm}
	\caption{A buggy transfer function (left) and its fixed version (right).}
	\label{prog:tran}
	\vspace{-4mm}
\end{figure}

\vspace{-1em}
\subsection{Translation to Solidity}\label{sec:translation}
\vspace{-0.5em}
In this section we discuss the problem of translating contracts from our programming language to Solidity, which is a widely-used language for programming contracts in Ethereum. 
There are two aspects in our language that are not automatically present in Solidity: (i)~the global clock, and (ii)~concurrent choices and payments by participants. 
We describe the two aspects below:
\begin{itemize} 

\item {\em Translation of Timing and the Clock.} The global clock can be modeled by the number of blocks in the blockchain. Solidity code is able to reference the blockchain. Given that a new block arrives roughly every $15$ to $20$ seconds, number of blocks that have been added to blockchain since the inception of the contract, or a constant multiple of it, can quantify passage of time.

\item {\em Translation of Concurrent Interactions.} Concurrent choices and payments can be implemented in Solidity using commitment schemes and digital signatures, which are standard tools in cryptography and cryptocurrencies. All parties first commit to their choice and then when they can no longer change it, unmask it. Commitment schemes can be extended to payments by requiring everyone to pay a fixed amount which is more than the value they are committing to and then returning the excess amount after unmasking.
\end{itemize} 
Hence contracts in our language can be automatically translated to Solidity.

\vspace{-1em}
\section{Bounded Analysis and Games} \label{sec:games} \label{SEC:GAMES}
\vspace{-0.5em}
Since smart contracts can be easily described in our programming language, 
and programs in our programming language can be translated to Solidity, 
the main aim to automatically compute values of contracts (i.e., compute 
guaranteed payoff for parties).
In this section, we introduce the bounded analysis problem for our programming 
language framework, and present concurrent games which is the underlying mathematical
framework for the bounded analysis problem.

\vspace{-1em}
\subsection{Bounded analysis}
\vspace{-0.5em}
As is standard in verification, we consider the bounded analysis problem, where 
the number of parties and the number of function calls are bounded.
In standard program analysis, bugs are often detected with a small number of processes, 
or a small number of context switches between concurrent threads. 
In the context of smart contracts, 
we analogously assume that the number of parties and function calls are bounded.

\smallskip\noindent{\em Contracts with bounded number of parties and function calls.}
Formally, a contract with bounded number of parties and function calls is as follows: 
\begin{itemize}
\item 
Let $C$ be a contract and $k \in \mathbb{N}$, we define $C_k$ as an equivalent contract that can have at most $k$ parties. 
This is achieved by letting $\parties = \{ \party_1, \party_2, \ldots, \party_k \}$ be the set of all possible ids in the contract. 
The set $\parties$ must contain all ids that are in the program source, therefore $k$ is at least the number of such ids. 
Note that this does not restrict that ids are controlled by unique users, and a real-life user can have several different ids.
We only restrict the analysis to bounded number of parties interacting with the smart contract.

\item To ensure runs are finite, number of function calls by each party is 
also bounded. 
Specifically, each party can call each function at most once during each time frame, i.e.~between two consecutive ticks of the clock. 
This closely resembles real-life contracts in which one's ability to call many functions is limited by the capacity of a block in the blockchain, 
given that the block must save all messages. For a more rigorous treatment see Appendix \ref{app:semantics}.
\end{itemize}

\vspace{-1em}
\subsection{Concurrent Games} 
\vspace{-0.45em}
The programming language framework we consider has interacting agents that act simultaneously,
and we have the program state. 
We present the mathematical framework of concurrent games, which are games played on finite 
state spaces with concurrent interaction between the players.

\smallskip\noindent\emph{Concurrent Game Structures.} \label{sec:concur_games}  \label{SEC:CONCUR_GAMES}
A concurrent two-player game structure is a tuple $G = (S, s_0, A, \Gamma_1, \Gamma_2, \delta)$, 
where $S$ is a finite set of states, $s_0 \in S$ is the start state, $A$ is a finite set of actions, 
$\Gamma_1, \Gamma_2 : S \rightarrow 2^A \setminus \emptyset$ such that $\Gamma_i$ assigns to each state $s \in S$, 
a non-empty set $\Gamma_i(s) \subseteq A$ of actions available to player $i$ at $s$, and finally 
$\delta: S \times A \times A \rightarrow S$ is a transition function that assigns to every state 
$s \in S$ and action pair $a_1 \in \Gamma_1(s), a_2 \in \Gamma_2(s)$ a successor state $\delta(s, a_1, a_2) \in S$.

\smallskip\noindent\emph{Plays and Histories.} The game starts at state $s_0$. At each state $s_i \in S$, player 1 chooses an action $a^i_1 \in \Gamma_1(s_i)$ and player 2 chooses an action $a_2^i \in \Gamma_2(s_i)$. The choices are made simultaneously and independently. The game subsequently transitions to the new state $s_{i+1} = \delta(s_i, a_1, a_2)$ and the same process continues. This leads to an infinite sequence of tuples $p = \left(s_i, a_1^i, a_2^i\right)_{i=0}^{\infty}$ which is called a {\em play} of the game. We denote the set of all plays by $\rah$. Every finite prefix $p[..r] := \left( (s_0, a_1^0, a_2^0), (s_1, a_1^1, a_2^1), \ldots, (s_{r}, a_1^{r}, a_2^{r}) \right)$ of a play is called a {\em history} and the set of all histories is denoted by $\Histories$. If $h = p[..r]$ is a history, we denote the last state appearing according to $h$, i.e.~$s_{r+1} = \delta(s_r, a_1^r, a_2^r)$, by $\last(h)$. We also define $p[..-1]$ as the empty history.

\smallskip\noindent\emph{Strategies and Mixed strategies.} 
A strategy is a recipe that describes for a player the action to play given the current game history. 
Formally, a strategy $\purestrategy_i$ for player $i$ is a function $\purestrategy_i : \Histories \rightarrow A$, 
such that $\purestrategy_i(h) \in \Gamma_i(\last(h))$. A pair $\purestrategy = (\purestrategy_1, \purestrategy_2)$ of strategies for the two players is called a strategy profile. Each such $\purestrategy$ induces a unique play.
A mixed strategy $\strategy_i : \Histories \rightarrow \Delta(A)$ for player $i$ given the history of the game. Intuitively, such a strategy suggests a distribution of actions to player $i$ at each step and then she plays one of them randomly according to that distribution. Of course it must be the case that $\Supp(\strategy_i(h)) \subseteq \Gamma_i(\last(h))$. 
A pair $\strategy = (\strategy_1, \strategy_2)$ of mixed strategies for the two players is called a mixed strategy profile.
Note that mixed strategies generalize strategies with randomization.
Every mixed strategy profile $\strategy = (\strategy_1, \strategy_2)$ induces a unique probability measure on the set of plays, 
which is denoted as $\Prob^{\strategy}[\cdot]$, and the associated expectation measure is denoted by $\E^{\strategy}[\cdot]$.

\smallskip\noindent \emph{State and History Utilities.} In a game structure $G$, a state utility function $u$ for player~1 is of the form $u: S \rightarrow \mathbb{R}$. Intuitively, this means that when the game enters state $s$, player~1 receives a reward of $u(s)$. State utilities can be extended to history utilities. We define the utility of a history to be the sum of utilities of all the states included in that history. Formally, if $h = \left(s_i, a_1^i, a_2^i \right)_{i=0}^{r}$, then $u(h) = \sum_{i=0}^{r} u(s_i)$. Given a play $p \in \rah$, we denote the utility of its prefix of length $\fhlen$ by $u_\fhlen(p)$.

\smallskip\noindent \emph{Games.} A game is a pair ($G$, $u$) where $G$ is a game structure and $u$ is a utility function for player~1. We assume that player~1 is trying to maximize $u$, while player~2's goal is to minimize it.

\smallskip\noindent \emph{Values.} The $\fhlen$-step finite-horizon value of a game $(G, u)$ is defined as
\begin{equation} \label{eq:gameval}
\gamevalue_\fhlen(G, u) := \sup_{\strategy_1} \inf_{\strategy_2} \E^{(\sigma_1, \sigma_2)} \left[ u_\fhlen(p)\right],
\end{equation}
where $\sigma_i$ iterates over all possible mixed strategies of player $i$. This models the fact that player~1 is trying to maximize the utility in the first $\fhlen$ steps of the run, while player~2 is minimizing it. 
The values of games can be computed using the value-iteration algorithm or dynamic programming, which is standard.
A formal treatment of the standard algorithms for games is presented in Appendix~\ref{app:games}.  

\begin{remark}
	Note that in (\ref{eq:gameval}), limiting player 2 to pure strategies does not change the value of the game. Hence, we can assume that player 2 is an arbitrarily powerful nondeterministic adversary and get the exact same results. 
\end{remark}

\vspace{-1.5em}
\subsection{Translating contracts to games} \label{sec:trans} \label{SEC:TRANS}
\vspace{-0.5em}
The translation from bounded smart contracts to games is straightforward, where the states of the concurrent game encodes
the states of the contract. 
Correspondences between objects in the contract and game are as follows: 
(a)~moves in contracts with actions in games; (b)~run prefixes in contracts with histories in games; 
(c)~runs in contracts with plays in games; and 
(d)~policies (resp., randomized policies) in contracts with strategies (resp., mixed strategies) in games. Note that since all runs of the bounded contract are finite and have a limited length, we can apply finite horizon analysis to the resulting game, where $\fhlen$ is the maximal length of a run in the contract. 
This gives us the following theorem:

\vspace{-0.7em}
\begin{theorem}[Correspondence] 
Given a bounded contract $C_k$ for a party $\party$ with objective $o$, a concurrent game can be constructed such that 
 value of this game, $\gamevalue_\fhlen(G, u)$, is equal to the value of the bounded contract, $\mathsf{V}(C_k, o, \party)$.
\end{theorem}
\vspace{-0.5em}

Details of the translation of smart contracts to games and proof of the theorem above is relegated to Appendix \ref{app:contogame}.
\vspace{-0.5em}
\begin{remark}
	Note that in standard programming languages where there is no interaction the underlying mathematical models are 
	graphs. In contrast, for the smart contracts programming languages we consider there are game theoretic interaction,
	and hence concurrent games on graphs are considered as the underlying mathematical model.
\end{remark}

\vspace{-1.5em}
\section{Abstraction for Quantitative Concurrent Games} \label{sec:abstract}
\vspace{-0.5em}
Abstraction is a key technique to handle large-scale systems.
In the previous section we described that smart contracts can be 
translated to games, but due to state-space explosion (since we allow
integer variables), the resulting state space of the game is huge.
Hence, we need techniques for abstraction, as well as refinement of abstraction, for 
concurrent games with quantitative utilities. 
In this section we present such abstraction refinement for quantitative
concurrent games, which is our main technical contribution in this paper. 
We prove soundness of our approach and its completeness in the limit. 
Then, we introduce a specific method of abstraction, called interval abstraction, 
which we apply to the games obtained from contracts and show that soundness and refinement are inherited from the general case. 
We also provide a heuristic for faster refining of interval abstractions for games obtained from contracts. 
\vspace{-1.3em}
\subsection{Abstraction for quantitative concurrent games}
\vspace{-0.5em}
Abstraction considers a partition of the state space, and reduces the number of states by taking each partition set as a state. 
In case of transition systems (or graphs) the standard technique is to consider existential (or universal) abstraction 
to define transitions between the partition sets. However, for game-theoretic interactions such abstraction ideas are not
enough.
We now describe the key intuition for abstraction in concurrent games with quantitative objectives and formalize it. We also provide a simple example for illustration.

\smallskip\noindent{\em Abstraction idea and key intuition.} 
In an abstraction the state space of the game $(G,u)$ is partitioned into several 
abstract states, where an abstract state represents a set of states of the original game.
Intuitively, an abstract state represents a set of similar states of the original game.
Given an abstraction our goal is to define two games that can provide lower and upper 
bound on the value of the original game.
This leads to the concepts of lower and upper abstraction.
\begin{itemize}

\item {\em Lower abstraction.} The lower abstraction $(G^{\down},u^{\down})$ represents 
a lower bound on the value. 
Intuitively, the utility is assigned as minimal utility among states in the partition, 
and when an action profile can lead to different abstract states, then the adversary, i.e.~player~2, chooses the transition. 

\item {\em Upper abstraction.} The upper abstraction $(G^{\up},u^{\up})$ represents 
an upper bound on the value. 
Intuitively, the utility is assigned as maximal utility among states in the partition, 
and when an action profile can lead to different abstract states, then player~1 
is chooses between the possible states. 
\end{itemize} 
Informally, the lower abstraction gives more power to the adversary, player~2, whereas 
the upper abstraction is favorable to player~1.

\smallskip\noindent{\em General abstraction for concurrent games.}
Given a game $(G, u)$ consisting of a game structure $G = (S, s_0, A, \Gamma_1, \Gamma_2, \delta)$ and a utility function $u$, and a partition $\partition$ of $S$, the lower and upper abstractions, $(G^\down = (S^\ab, s_0^\ab, A^\ab, \Gamma_1^\down, \Gamma_2^\down, \delta^\down), u^\down)$ and $(G^\up = (S^\ab, s_0^\ab, A^\ab, \Gamma_1^\up, \Gamma_2^\up, \delta^\up), u^\up)$, of $(G, u)$ with respect to $\partition$ are defined as: 
\begin{itemize}
	\item $S^\ab = \partition \cup \Dummy$, where $\Dummy = \partition \times A \times A$ is a set of dummy states for giving more power to one of the players. Members of $S^\ab$ are called abstracted states.
	\item The start state of $G$ is in the start state of $G^\up$ and $G^\down$, i.e.~$s_0 \in s_0^\ab \in \partition$.
	\item $A^\ab = A \cup \partition$. Each action in abstracted games either corresponds to an action in the original game or to a choice of the next state.
	\item If two states $s_1, s_2 \in S$, are in the same abstracted state $s^\ab \in \partition$, then they must have the same set of available actions for both players, i.e.~$\Gamma_1(s_1) = \Gamma_1(s_2)$ and $\Gamma_2(s_1) = \Gamma_2(s_2)$. Moreover, $s^\ab$ inherits these action sets. Formally, $\Gamma_1^\down(s^\ab) = \Gamma_1^\up(s^\ab) = \Gamma_1(s_1) = \Gamma_1(s_2)$ and $\Gamma_2^\down(s^\ab) = \Gamma_2^\up(s^\ab) = \Gamma_2(s_1) = \Gamma_2(s_2)$. 
	\item For all $\pi \in \partition$ and $a_1 \in \Gamma_1^\down(\pi)$ and $a_2 \in \Gamma_2^\down(\pi)$, we have $\delta^\down(\pi, a_1, a_2) = (\pi, a_1, a_2) \in \Dummy$. Similarly for $a_1 \in \Gamma_1^\up(\pi)$ and $a_2 \in \Gamma_2^\up(\pi)$, $\delta^\up (\pi, a_1, a_2) = (\pi, a_1, a_2) \in \Dummy$. This means that all transitions from abstract states in $\partition$ go to the corresponding dummy abstract state in $\Dummy$.  
	\item If $\dummy = (\pi, a_1, a_2) \in \Dummy$ is a dummy abstract state, then let $X_\dummy = \{ \pi' \in \partition ~~~ \vert ~~~ \exists ~~~ s \in \pi \quad \delta(s, a_1, a_2) \in \pi' \}$ be the set of all partition sets that can be reached from $\pi$ by $a_1, a_2$ in $G$. Then in $G^\down$, $\Gamma^\down_1(\dummy)$ is a singleton, i.e., player~1 has no choice, and $\Gamma^\down_2(\dummy) = X_d$, i.e., player~2 can choose which abstract state is the next. Conversely, in $G^\up$, $\Gamma^\up_2(d)$ is a singleton and player~2 has no choice, while $\Gamma^\up_1(d) = X_d$ and player~1 chooses the next abstract state.
	\item In line with the previous point, $\delta^\down(d, a_1, a_2) = a_2$ and $\delta^\up(d, a_1, a_2) = a_1$ for all $d \in \Dummy$ and available actions $a_1$ and $a_2$.  
	\item We have $u^\down(s^\ab) = \min_{s \in s^\ab} \{u(s)\}$ and $u^\up(s^\ab) = \max_{s \in s^\ab} \{u(s)\}$. The utility of a non-dummy abstracted state in $G^\down$, resp. $G^\up$, is the minimal, resp. maximal, utility among the normal states included in it. Also, for each dummy state $\dummy \in \Dummy$, we have $u^\down(\dummy) = u^\up(\dummy) = 0$.
\end{itemize}
Given a partition $\partition$ of $S$, either (i)~there is no lower or upper abstraction corresponding to it because it puts states with different sets of available actions together; or (ii)~there is a unique lower and upper abstraction pair. 
Hence we will refer to the unique abstracted pair of games by specifying $\partition$ only. 

\begin{remark}\label{rem:dummy}
Dummy states are introduced for conceptual clarity in explaining the ideas because in lower abstraction all choices
are assigned to player~2 and upper abstraction to player~1.
However, in practice, there is no need to create them, as the choices can be allowed to the respective players in the predecessor state. 
\end{remark}

\smallskip\noindent{\em Example.} Figure~\ref{pic:examplefix} (left) shows a concurrent game with $(G, u)$ with $4$ states. The utilities are denoted in red. The edges correspond to transitions in $\delta$ and each edge is labeled with its corresponding action pair. Here $A = \{ \a, \b \}$, $\Gamma_1(s_0) = \Gamma_2(s_0) = \Gamma_2(s_1) = \Gamma_1(s_2) = \Gamma_2(s_2) = \Gamma_2(s_3) = A$ and $\Gamma_1(s_1) = \Gamma_1(s_3) = \{ \a \}$.
Given that action sets for $s_0$ and $s_2$ are equal, we can create abstracted games using the partition $\partition = \{ \pi_0, \pi_1, \pi_2 \}$ where $\pi_1 = \{s_0, s_2\}$ and other sets are singletons. The resulting game structure is depicted in Figure~\ref{pic:examplefix} (center). Dummy states are shown by circles and whenever a play reaches a dummy state in $G^\down$, player~2 chooses which red edge should be taken. Conversely, in $G^\up$ player~1 makes this choice. Also, $u^\up(\pi_0) = \max \{u(s_0), u(s_2)\} = 10, u^\down(\pi_0) = \min \{u(s_0), u(s_2)\} = 0$ and $u^\up(\pi_1) u^\down(\pi_1) = u(s_1) = 10, u^\up(\pi_2) = u^\down(\pi_2) = u(s_3) = 0$.
The final abstracted $G^\down$ of the example above, without dummy states, is given in Figure~\ref{pic:examplefix} (right).

\begin{figure}[!htb]
	\vspace{-4mm}
	\hspace{-1cm}
	\includegraphics[scale=0.11]{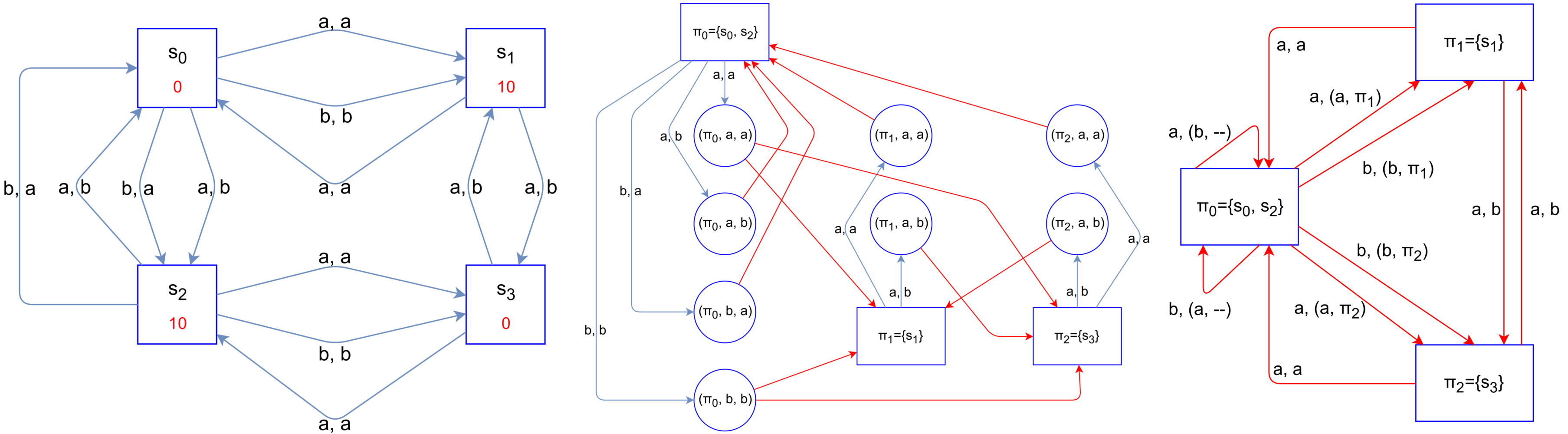}
	\vspace{-2mm}
	\caption{An example concurrent game (left), abstraction process (center) and the corresponding $G^\down$ without dummy states (right).}
	\label{pic:examplefix}
	\vspace{-5mm}
\end{figure}

\vspace{-1em}
\subsection{Abstraction: soundness, refinement, and completeness in limit}
\vspace{-0.5em}
For an abstraction we need to prove three key properties: (a)~soundness, 
(b)~refinement of the abstraction, and (c)~completeness in the limit.
The intuitive description is as follows:
(a)~soundeness requires that the value of the games is between the value of the 
lower and upper abstraction; 
(b)~refinement requires that if the partition is refined, then the values of 
lower and upper abstraction becomes closer;
and (c)~completeness requires that if the partitions are refined enough, then the
value of the original game can be approximated.
We present and prove each of these results below. 

\vspace{-1em}
\subsubsection{Soundness.} 
Soundness means that when we apply abstraction, value of the original game must lie between values of the lower and upper abstractions. Intuitively, this means abstractions must provide us with some interval containing the value of the game. We expect the value of $(G^\down, u^\down)$ to be less than or equal to the value of the original game because in $(G^\down, u^\down)$, the utilities are less than in $(G, u)$ and player~2 has more power, given that she can choose which transition to take. Conversely, we expect $(G^\up, u^\up)$ to have a higher value than $(G, u)$.

\smallskip\noindent{\em Formal requirement for Soundness.}
An abstraction of a game $(G, u)$ leading to abstraction pair $(G^\up, u^\up), (G^\down, u^\down)$ is sound if for every $\fhlen$, 
$$
\gamevalue_{2 \fhlen}(G^\down, u^\down) \leq \gamevalue_{\fhlen}(G, u) \leq \gamevalue_{2 \fhlen}(G^\up, u^\up).
$$
The factor $2$ in the inequalities above is due to the fact that each transition in the original game is modeled by two transitions in abstracted games, one to a dummy state and a second one out of that dummy state.

We now formally prove our soundness result. The main intuition in this proof is that letting player~1 get the minimal reward in each partition set when she reaches any state of the set, and allowing player~2 to choose the resulting state among all possibilities cannot possibly be in player~1's favor and increase her utility. Similarly, doing the opposite thing by letting her get the maximal reward and choose the transition cannot possibly decrease her utility.

\begin{theorem}[Soundness] \label{thm:soundness}
Given a game $(G, u)$ and a partition $\partition$ of its state space, 
if $G^\up$ and $G^\down$ exist, then the abstraction is sound, i.e.~ for all $\fhlen$, it is the case that $\gamevalue_{2 \fhlen}(G^\down, u^\down) \leq \gamevalue_{\fhlen}(G, u) \leq \gamevalue_{2 \fhlen}(G^\up, u^\up)$.
\end{theorem} 
\begin{proof}
	We prove the first inequality, the second one can be done similarly. 
	
	For a mixed strategy $\sigma_1^\down$ for player~1 in $G^\down$, let $v(\sigma_1^\down) := \inf_{\sigma_2^\down} \E^{(\sigma_1^\down, \sigma_2^\down)} [u^\down_{2 \fhlen} (p)] $ be the guaranteed value of the game if player~1 plays $\sigma_1^\down$. We say that $\sigma^\down_2$ is a best response to $\sigma^\down_1$ if playing $\sigma^\down_1, \sigma^\down_2$ leads to a total utility of $v(\sigma^\down_1)$. We define $v(\sigma_1)$ and best responses in $G$ analogously.
	
	Let $\sigma^\down_1$ be a strategy for player~1 in $G^\down$, such that $v(\sigma^\down_1) = \sup_{\varsigma^\down_1} v(\varsigma^\down_1) = \gamevalue_{2\fhlen} (G^\down, u^\down)$. Such a strategy exists because the set of all strategies for player~1 is compact and $\gamevalue_{2\fhlen}$ is continuous. Let $\sigma_1$ be a strategy for player~1 in $G$ that follows $\sigma^\down_1$, i.e.~looks at the histories of $G$ as histories of $G^\down$ and assigns to each history of $G$ the action that $\sigma^\down_1$ assigns to the corresponding history of $G^\down$. Then, let $\sigma_2$ be a best response to $\sigma_1$ in $G$ and $\sigma^\down_2$ a strategy in $G^\down$ that follows $\sigma_2$ against $\sigma^\down_1$, i.e.~chooses actions in non-dummy states in accordance with actions chosen by $\sigma_2$ and actions in dummy states in accordance with transitions of the play of $(\sigma_1, \sigma_2)$. Intuitively, $\sigma^\down_2$ is player~2's best strategy if she does not use her additional ability of choosing the next state in $G^\down$. It is evident by construction that $\E^{(\sigma_1^\down, \sigma_2^\down)} [u^\down_{2 \fhlen} (p)] \leq \E^{(\sigma_1, \sigma_2)} [u_{\fhlen} (p)]$, because paths in $G^\down$ according to these strategies go through a sequence of partition sets that correspond exactly to the sequence of states that are visited in $G$ and the utility of each such partition set is defined to be less than or equal to the utility of each of its states. Therefore $v(\sigma^\down_1) \leq v(\sigma_1)$, given that $\sigma_2$ was a best response. This means that $\gamevalue_{2 \fhlen}(G^\down, u^\down) = v(\sigma_1^\down) \leq v(\sigma_1) \leq \gamevalue_{\fhlen}(G, u)$.
\end{proof}

\vspace{-1em}
\subsubsection{Refinement.}
We say that a partition $\partition_2$ is a refinement of a partition $\partition_1$, and write $\partition_2 \refines \partition_1$, if every $\pi \in \partition_1$ is a union of several $\pi_i$'s in $\partition_2$, i.e.~$\pi = \bigcup_{i \in \mathcal{I}} \pi_{i}$ and for all $i \in \mathcal{I}$, $\pi_{i} \in \partition_2$. Intuitively, this means that $\partition_2$ is obtained by further subdividing the partition sets in $\partition_1$. It is easy to check that $\refines$ is a partial order over partitions. We expect that if $\partition_2 \refines \partition_1$, then the abstracted games resulting from $\partition_2$ give a better approximation of the value of the original game in comparison with abstracted games resulting from $\partition_1$. This is called the refinement property.

\smallskip\noindent{\em Formal requirement for the Refinement Property.} Two abstractions of a game $(G, u)$ using two partitions $\partition_1, \partition_2$, such that $\partition_2 \refines \partition_1$, and leading to abstracted games $(G_i^\up, u_i^\up), (G_i^\down, u_i^\down)$ corresponding to each $\partition_i$ satisfy the refinement property if for every $\fhlen$,
	$$
	\gamevalue_{2 \fhlen}(G^\down_1, u^\down_1) \leq \gamevalue_{2\fhlen}(G^\down_2, u^\down_2) \leq \gamevalue_{2 \fhlen}(G^\up_2, u^{\up}_2) \leq \gamevalue_{2 \fhlen}(G^\up_1, u^\up_1).
	$$

We now prove that any two abstractions with $\partition_2 \refines \partition_1$ satisfy this property.

\begin{theorem}[Refinement Property] \label{thm:refinement}
	Let $\partition_2 \refines \partition_1$ be two partitions of the state space of a game $(G, u)$, then the abstractions corresponding to $\partition_1, \partition_2$ satisfy the refinement property.
\end{theorem}
\begin{proof}
	Note that $\partition_2 \refines \partition_1$, so $\partition_1$ can itself be considered as a partition of $\partition_2$ and one can then define an abstracted pair of games $(G_2^{\down\down}, u_2^{\down\down})$ and $(G_2^{\down\up}, u_2^{\down\up})$ on $(G_2^\down, u_2^\down)$ with respect to $\partition_1$. Strategies and paths in $G_2^{\down\down}$ are in natural bijection with strategies and paths in $G_1^\down$ and the bijection preserves utility. Therefore, using the Soundness theorem above, we have $\gamevalue_{2\fhlen}(G_1^\down, u_1^\down) = \gamevalue_{4 \fhlen}(G_2^{\down \down}, u_2^{\down \down}) \leq \gamevalue_{2 \fhlen} (G_2^\down, u_2^\down)$. The other inequality is proven by a similar argument.
\end{proof}
\vspace{-1em}
\subsubsection{Completeness in the limit.}

We say that an abstraction is complete in the limit, if by refining it enough the values of upper and lower abstractions get as close together as desired. Equivalently, this means that if we want to approximate the value of the original game within some predefined threshold of error, we can do so by repeatedly refining the abstraction. 

\smallskip\noindent{\em Formal requirement for Completeness in the limit.} 
Given a game $(G, u)$, a fixed finite-horizon $\fhlen$ and an abstracted game pair corresponding to a partition $\partition_1$, the abstraction is said to be complete in the limit, if for every $\epsilon \geq 0$ there exists $\partition_2 \refines \partition_1$, such that if $(G^\down_2, u^\down_2), (G^\up_2, u^\up_2)$ are the abstracted games corresponding to $\partition_2$, then $\gamevalue_\fhlen(G_2^\up, u_2^\up) - \gamevalue_\fhlen(G_2^\down, u_2^\down) \leq \epsilon$.

\begin{theorem}[Completeness in the Limit] \label{thm:completeness}
	Every abstraction on a game $(G, u)$ using a partition $\partition$ is complete in the limit for every value of $\fhlen$.
\end{theorem} 
\begin{proof}
	Consider the unit partition $\partition_* = \{ \{s\} \vert s \in S \}$, where $S$ is the set of states of $G$. For every partition $\partition$, we have $\partition_* \refines \partition$. On the other hand, the upper and lower abstracted games of $(G, u)$ with respect to $\partition_*$ are simply the same game as $(G, u)$ except that each transition now first goes to a dummy state and then comes out of it without any choice for any of the players. Therefore their values are the same and equal to the value of $(G, u)$.
\end{proof}

\vspace{-1em}
\subsection{Interval Abstraction}
\vspace{-0.5em}
 In this section, we turn our focus to games obtained from contracts and provide a specific method of abstraction that can be applied to them.

\smallskip\noindent{\em Intuitive Overview.}
Let $(G, u)$ be a concurrent game obtained from a contract as in the Section \ref{sec:trans}. Then the states of $G$, other than the unique dummy state, correspond to states of the contract $C_k$. Hence, they are of the form $s = (t, b, l, val, p)$, where $t$ is the time, $b$ the contract balance, $l$ is a label, $p$ is the party calling the current function and $val$ is a valuation. In an abstraction, one cannot put states with different times or labels or callers together, because they might have different moves and hence different action sets in the corresponding game. The main idea in interval abstraction is to break the states according to intervals over their balance and valuations. We can then refine the abstraction by making the intervals smaller. We now formalize this concept.

\smallskip \noindent{\em Objects.} Given a contract $C_k$, let $\objects$ be the set of all objects that can have an integral value in a state $s$ of the contract. This consists of the contract balance, numeric variables and $m[\party]$'s where $m$ is a map variable and $\party$ is a party. More precisely, $\objects = \{ \balance \} \cup N \cup \{ m[\party] \vert m \in M, \party \in \parties \}$ where $\balance$ denotes the balance. For an $o \in \objects$, the value assigned to $o$ at state $s$ is denoted by $o_s$.

\smallskip \noindent{\em Interval Partition.} Let $C_k$ be a contract and $(G, u)$ its corresponding game. A partition $\partition$ of the state space of $G$ is called an interval partition if:

\begin{itemize}
	\item The dummy state is put in a singleton set $\pi_\dummy$.
	\item Each $\pi \in \partition$ except $\pi_\dummy$ has associated values, $t_\pi, l_\pi, \party_\pi$ and for each $o \in \objects$, $\overline{o}_\pi, \underline{o}_\pi$, such that: 
	$$
	\pi = \{ s \in S \vert s=(t_\pi, b, l_\pi, \val, \party_\pi) \text{ and for all } o \in \objects,~~\underline{o}_\pi \leq s_o \leq \overline{o}_\pi \}.
	$$ 
	
	Basically, each partition set includes states with the same time, label and caller in which the value of every object $o$ is in an interval $[\underline{o}_\pi, \overline{o}_\pi]$.
	\end{itemize}
We call an abstraction using an interval partition, an interval abstraction.

\smallskip \noindent{\em Refinement Heuristic.} We can start with big intervals and continually break them into smaller ones to get refined abstractions and a finer approximation of the game value. We use the following heuristic to choose which intervals to break: Assume that the current abstracted pair of games are $(G^\down, u^\down)$ and $(G^\up, u^\up)$ corresponding to an interval partition $\partition$. Let $\dummy = (\pi_\dummy, a_1, a_2)$ be a dummy state in $G^\up$ and define the skewness of $\dummy$ as $\gamevalue(G^\up_\dummy, u^\up) - \gamevalue(G^\down_\dummy, u^\down)$. Intuitively, skewness of $\dummy$ is a measure of how different the outcomes of the games $G^\up$ and $G^\down$ are, from the point when they have reached $\dummy$. Take a label $l$ with maximal average skewness among its corresponding dummy states and cut all non-unit intervals of it in more parts to get a new partition $\partition'$. Continue the same process until the approximation is as precise as desired. 
Intuitively, it tries to refine parts of the abstraction that show the most disparity between $G^\down$ and $G^\up$ with the aim to bring their values closer. Our experiments show its effectiveness.

\smallskip \noindent{\em Soundness and Completeness in the limit.} If we restrict our attention to interval abstractions, soundness is inherited from general abstractions and completeness in the limit holds because $\partition_*$ is an interval partition. Therefore, using interval abstractions is both sound and complete in the limit.

\smallskip \noindent{\em Interval Refinement.} An interval partition $\partition'$ is interval refinement of a given interval partition $\partition$ if $\partition' \refines \partition$. Refinement property is inherited from general abstractions. This intuitively means that $\partition'$ is obtained by breaking the intervals in some sets of $\partition$ into smaller intervals.

\smallskip\noindent{\em Conclusion.}
We devised a sound abstraction-refinement method for approximating values of contracts. Our method is also complete in the limit. It begins by converting the contract to a game, then applies interval abstraction to the resulting game and repeatedly refines the abstraction using a heuristic until the desired precision is reached.

%

\vspace{-1em}
\section{Experimental Results}
\label{sec:exp} \label{SEC:EXP}
\vspace{-0.5em}

\subsection{Implementation and Optimizations} \label{SEC:IMPL} \label{sec:impl}
\vspace{-0.5em}
The state-space of the games corresponding to the smart contracts is huge.
Hence the original game corresponding to the contract is computationally too
expensive to construct. 
Therefore, we do not first construct the game and then apply abstraction, 
instead we first apply the interval abstraction, and construct
the lower and upper abstraction and compute values in them. We optimized our implementation by removing dummy states and exploiting acyclicity using backward-induction. Details of experiments are provided in Appendix~\ref{app:exper}.

\vspace{-1em}
\subsection{Experimental Results} \label{sec:res} \label{SEC:RES}
\vspace{-0.5em}
In this section we present our experimental results (Table~\ref{tab:allres}) for the five examples
mentioned in Section~\ref{sec:examples}. 
In each of the examples, the original game is quite large, and the size of
the state space is calculated without creating them.
In our experimental results we show the abstracted game size, 
the refinement of games to larger sizes, and how the lower and upper bound 
on the values change. 
We used an Ubuntu machine with 3.2GHz Intel i7-5600U CPU and 12GB RAM.

\begin{table}[!!htbp]
	\centering
	\vspace{-3mm}
	\resizebox{12.6cm}{!}{

	\begin{tabular}{|c|c|lcr|c|c|lcr|c|}
		\hline
		\multicolumn{11}{|c|}{\textbf{Rock-Paper-Scissors}}                                                                                                                                                                                                                                    \\ \hline
		 Size                          & \multicolumn{10}{c|}{Abstractions}                                                                                                                                                                                                   \\ \hline
		\multirow{5}{*}{$> 2.5 \cdot 10^{14}$} & \multicolumn{5}{c|}{Correct Program}                                                    & \multicolumn{5}{c|}{Buggy Variant}                                                                                                         \\ \cline{2-11} 
		& states   & $[l$ & $,$ & $u]$ & time   & states                    & $[l$ & $,$                  & $u]$ & time                    \\ \cline{2-11} 
		& $19440$  & $[0.00$                         & $,$ & $10.00]$                    & $367$  & $25200$                   & $[0.00$                         & $,$                  & $10.00]$                    & $402$                   \\ \cline{2-11} 
		& $135945$ & $[1.47$                         & $,$ & $6.10]$                     & $2644$ & \multirow{2}{*}{$258345$} & \multirow{2}{*}{$[8.01$}        & \multirow{2}{*}{$,$} & \multirow{2}{*}{$10.00]$}   & \multirow{2}{*}{$4815$} \\ \cline{2-6}
		& $252450$ & $[1.83$                         & $,$ & $5.59${]}                   & $3381$ &                           &                                 &                      &                             &                         \\ \hline \hline

		\multicolumn{11}{|c|}{\textbf{Auction}}                                                                                                                                                                                                                          \\ \hline
		Size                                  & \multicolumn{10}{c|}{Abstractions}                                                                                                                                                                                       \\ \hline
		\multirow{5}{*}{$>5.2 \cdot 10^{14}$} & \multicolumn{5}{c|}{Correct Program}                                                                        & \multicolumn{5}{c|}{Buggy Variant}                                                                         \\ \cline{2-11} 
		& states   & $[l$ & \multicolumn{1}{l}{$,$} & $u]$ & time   & states   & $[l$ & \multicolumn{1}{l}{$,$} & $u]$ & time  \\ \cline{2-11} 
		& $3360$   & $[0$                            & $,$                     & $1000]$                     & $68$   & $2880$   & $[0$                            & $,$                     & $1000]$                     & $38$  \\ \cline{2-11} 
		& $22560$  & $[0$                            & $,$                     & $282]$                      & $406$  & $27360$  & $[565$                          & $,$                     & $1000]$                     & $552$ \\ \cline{2-11} 
		& $272160$ & $[0$                            & $,$                     & $227]$                      & $4237$ & $233280$ & $[748$                          & $,$                     & $1000]$                     & 3780  \\ \hline \hline

		\multicolumn{11}{|c|}{\textbf{Lottery}}                                                                                                                                                                                                                                                                                          \\ \hline
		Size                     & \multicolumn{10}{c|}{Abstractions}                                                                                                                                                                                                                                                          \\ \hline
		\multirow{5}{*}{$>2.5 \cdot 10^8$} & \multicolumn{5}{c|}{Correct Program}                                                                                                         & \multicolumn{5}{c|}{Buggy Variant}                                                                                                           \\ \cline{2-11} 
		& states                     & $[l$ & $,$                  & $u]$ & time                     & states                     & $[l$ & $,$                  & $u]$ & time                     \\ \cline{2-11} 
		& $1539$                     & $[-1$                           & $,$                  & $1]$                        & $17$                     & $1701$                     & $[-1$                           & $,$                  & $1]$                        & $22$                     \\ \cline{2-11} 
		& \multirow{2}{*}{$2457600$} & \multirow{2}{*}{$[0$}           & \multirow{2}{*}{$,$} & \multirow{2}{*}{$0]$}       & \multirow{2}{*}{$13839$} & \multirow{2}{*}{$2457600$} & \multirow{2}{*}{$[-1$}          & \multirow{2}{*}{$,$} & \multirow{2}{*}{$-1]$}      & \multirow{2}{*}{$13244$} \\
		&                            &                                 &                      &                             &                          &                            &                                 &                      &                             &                          \\ \hline

\end{tabular}

\begin{tabular}{|c|c|lcr|c|c|lcr|c|}
	\hline
	\multicolumn{11}{|c|}{\textbf{Sale}}                                                                                                                                                                                                                              \\ \hline
	Size                                  & \multicolumn{10}{c|}{Abstractions}                                                                                                                                                                                        \\ \hline
	\multirow{5}{*}{$>4.6 \cdot 10^{22}$} & \multicolumn{5}{c|}{Correct Program}                                                                        & \multicolumn{5}{c|}{Buggy Variant}                                                                          \\ \cline{2-11} 
	& states   & $[l$ & \multicolumn{1}{l}{$,$} & $u]$ & time   & states   & $[l$ & \multicolumn{1}{l}{$,$} & $u]$ & time   \\ \cline{2-11} 
	& $17010$  & $[0$                            & $,$                     & $2000]$                     & $226$  & $17010$  & $[0$                            & $,$                     & $2000]$                     & $275$  \\ \cline{2-11} 
	& $75762$  & $[723$                          & $,$                     & $1472]$                     & $1241$ & $81202$  & $[1167$                         & $,$                     & $2000]$                     & $1733$ \\ \cline{2-11} 
	& $131250$ & $[792$                          & $,$                     & $1260]$                     & $2872$ & $124178$ & $[1741$                         & $,$                     & $2000]$                     & 2818   \\ \hline \hline

	\multicolumn{11}{|c|}{\textbf{Transfer}}                                                                                                                                                                                                                           \\ \hline
	Size               & \multicolumn{10}{c|}{Abstractions}                                                                                                                                                                                                   \\ \hline
	\multirow{5}{*}{$>10^{23}$} & \multicolumn{5}{c|}{Correct Program}                                                    & \multicolumn{5}{c|}{Buggy Variant}                                                                                                         \\ \cline{2-11} 
	& states   & $[l$ & $,$ & $u]$ & time   & states                    & $[l$ & $,$                  & $u]$ & time                    \\ \cline{2-11} 
	& $1040$   & $[0$                            & $,$ & $2000]$                     & $20$   & $6561$                    & $[0$                            & $,$                  & $2000]$                     & $237$                   \\ \cline{2-11} 
	& $32880$  & $[844$                          & $,$ & $1793]$                     & $562$  & \multirow{2}{*}{$131520$} & \multirow{2}{*}{$[1716$}        & \multirow{2}{*}{$,$} & \multirow{2}{*}{$2000]$}    & \multirow{2}{*}{$3979$} \\ \cline{2-6}
	& $148311$ & $[903$                          & $,$ & $1352${]}                   & $3740$ &                           &                                 &                      &                             &                         \\ \hline
\end{tabular}

}
		\vspace{1mm}
		\caption{Experimental results for the contracts and their buggy counterparts. $l:=\gamevalue(G^\down, u^\down)$ denotes the lower value and $u := \gamevalue(G^\up, u^\up)$ is the upper value. Running times are reported in seconds.}
		\label{tab:allres}
		\vspace{-8mm}
\end{table}

	

	

	

	

\smallskip\noindent{\em Interpretation of the experimental results.}
Our results demonstrate the effectiveness of our approach 
in automatically approximating values of large games
and real-world smart contracts. Concretely, the following points are shown:

\begin{itemize}
	\item {\em Refinement Property.} By repeatedly refining the abstractions, values of lower and upper abstractions get closer at the expense of a larger state space.
	\item {\em Distinguishing Correct and Buggy Programs.} Values of the lower and upper abstractions provide an approximation interval containing the contract value. These intervals shrink with refinement until the intervals for correct and buggy programs become disjoint and distinguishable. 
	\item {\em Bug Detection.} One can anticipate a sensible value for the contract, and an approximation interval not containing the value shows a bug. For example, in token sale, the objective (number of tokens sold) is at most $1000$, while results show the buggy program has a value between $1741$ and $2000$.
	\item {\em Quantification of Economic Consequences.} Abstracted game values can also be seen as a method to quantify and find limits to the economic gain or loss of a party. For example, our results show that if the buggy auction contract is deployed, a party can potentially gain no more than $1000$ units from it.
	
\end{itemize}

\vspace{-5mm}
\section{Comparison with Related Work}\label{sec:related} \label{app:compare} \label{sec:comparison_new}
\vspace{-2mm}
\smallskip\noindent{\em Blockchain security analysis.}
The first security analysis of Bitcoin protocol was done by Nakamoto~\cite{nakamoto2008bitcoin} who showed resilience of the blockchain against double-spending. A stateful analysis was done by Sapirshtein et al.~\cite{sapirshtein2015optimal} and by Sompolinsky and Zohar~\cite{SompolinskyZ16}
in which states of the blockchain were considered. It was done using MDPs where only the attacker decides on her actions and the victim follows a predefined protocol.
Our paper is the first work that is using two-player and concurrent games to analyze contracts and the first to use stateful analysis on arbitrary smart contracts, rather than a specific protocol.

\smallskip\noindent{\em Smart contract security.}
Delmolino et al.~\cite{delmolino2015step} held a contract programming workshop and showed that even simple contracts can contain incentive misalignment bugs.
Luu et al.~\cite{loiluu} introduced a symbolic model checker with which they could detect specific erroneous patterns.
However the use of model checker cannot be extended to game-theoretic analysis.
Bhargavan et al.~\cite{bhargavan2016formal} translated solidity programs to $F^*$ and then used standard verification tools to detect vulnerable code patterns.
See \cite{atzei2016survey} for a survey of the known causes for Solidity bugs that result in security vulnerabilities.

\smallskip\noindent{\em Games and verification.}
Abstraction for concurrent games has been considered wrt qualitative temporal objectives~\cite{Church62,PR89,de2004three,alur1998alternating}.
Several works considered concurrent games with only deterministic (pure) strategies~\cite{HenzingerJM03,SAS2000,de2004three}. Concurrent games with pure strategies are extremely restrictive and effectively similar to turn-based games. The min-max theorem (determinacy) does not hold for them even in special cases of one-shot games or games with qualitative objectives.

Quantitative analysis with games has also been considered~\cite{BCHJ09,CCHRS11,chatterjee2015qualitative}.
However these approaches either consider games without concurrent interactions or do not 
consider any abstraction-refinement.
A quantitative abstraction-refinement framework has been considered in~\cite{CernyHR13};
however, there is no game-theoretic interaction.
Abstraction-refinement for games has also been considered 
for counter-example guided control and planning~\cite{ChatterjeeHJM05,HenzingerJM03};
however, these works do not consider games with concurrent interaction, nor quantitative objectives.
Moreover,~\cite{ChatterjeeHJM05,HenzingerJM03}  start with a finite-state model without variables, and interval 
abstraction is not applicable to these game-theoretic frameworks.
In contrast, our technical contribution is an abstraction-refinement 
approach for quantitative games and its application to analysis of 
smart contracts. A recent work~\cite{prevpaper} also considers quantitative games for analyzing attacks in cryptocurrencies, but it focuses on long-term mean-payoff analysis, rather than finite-horizon analysis, and is not able to automatically synthesize the games from the protocols. 

\smallskip\noindent{\em Formal methods in security.}
There is a huge body of work on program analysis for security;
see~\cite{SabelfeldM03,Abadi12} for survey.
Formal methods are used to create safe programming languages (e.g.,~\cite{fuchs2009scandroid,SabelfeldM03}) and to define new logics that can express security properties (e.g.,~\cite{burrows1989logic,flac,ArdenLM15}).
They are also used to automatically verify security and cryptographic protocols, e.g.,~\cite{abadi2000reconciling, blanchet2008automated} and~\cite{avalle2014formal} for a survey.
However, all of these works aimed to formalize qualitative properties such as privacy violation and information leakage.
To the best of our knowledge, our framework is the first attempt to use
formal methods as a tool for reasoning about monetary loses and identifying them as security errors.

\smallskip\noindent{\em Bounded model checking.}
Bounded Model Checking
(BMC), was first proposed by Biere et al. in 1999~\cite{biere1999symbolic}.
The basic idea in BMC is to search for a counterexample in executions whose length
is bounded by some integer $k$. If no bug is found then one increases $k$ until either a bug
is found, the problem becomes intractable, or some pre-known upper bound is reached.

\smallskip\noindent{\em Interval abstraction.}
The first infinite abstract domain was introduced in~\cite{cousot1977static}. This was later used
to prove that infinite abstract domains can lead to effective static
analysis for a given programming language~\cite{cousot1992comparing}.
However, none of the standard techniques is applicable to game analysis.

\vspace{-4mm}
\section{Conclusion}
\vspace{-2.5mm}
\label{sec:conclusion}
In this work we present a programming language for smart contracts, and an abstraction-refinement approach for quantitative concurrent games 
to automatically analyze (i.e., compute worst-case guaranteed utilities of)
such contracts. This is the first time a quantitative stateful game-theoretic framework is studied for formal analysis of smart contracts. There are several interesting directions of future work.
First, we present interval-based abstraction techniques for such games, and 
whether different abstraction techniques can lead to more scalability or 
other classes of contracts is an interesting direction of future work.
Second, since we consider worst-case guarantees, the games we obtain are 
two-player zero-sum games. 
The extension to study multiplayer games and compute values for rational agents
is another interesting direction of future work.
Finally, in this work we do not consider interaction between smart contracts,
and an extension to encompass such study will be a subject of its own.

\newpage
\bibliographystyle{plain}
\bibliography{krish1,krish2,yaron,amir}

\newpage
\appendix
\newpage
\section{Appendix to Section \ref{sec:prog_lang}: Contracts Programming Language}
\subsection{Appendix to Section \ref{sec:syntax}: Formal Syntax}\label{app:syntax}
The following grammar formally defines the syntax of our language, the non-terminal ``Contract'' is considered to be the start symbol and non-terminals and keywords are enclosed in quotation marks. We usually refrain from putting unnecessary parentheses or braces in the examples, even though they are formally part of this grammar. Also, $\varepsilon$ denotes the empty string.

\begin{grammar}
	[(colon){$\rightarrow$}]
	[(semicolon){ $|$}]
	[(period){ $~$}]
	[(quote){`}{'}]
	Contract : "contract" ContractName "\{" VariableList . FunctionList "\}"
	
	VariableList : VariableDefinition ; VariableList . VariableDefinition
	
	VariableDefinition : NumericDefinition ; MapDefinition ; IdDefinition
	
	NumericDefinition : "numeric" . NumericVariableName . "" Integer "," Integer "]" "=" Integer ";"
	
	MapDefinition : "map" . MapVariableName . "[" Integer "," Integer "]" "=" Integer ";"
	
	IdDefinition : "id" . IdVariableName . "=" Integer ";" ; "id" . IdVariableName . "=" "null" ";"
	
	FunctionList : Function ; Function . FunctionList
	
	Function : OnePartyFunction ; MultiPartyFunction
	
	OnePartyFunction : FunctionHeader . "(" . ParameterList . ")" "\{" CommandList "\}"
	
	MultiPartyFunction : FunctionHeader  "("  ConcurrentParamList ")" "\{" CommandList"\}"
	
	FunctionHeader : "function" FunctionName "[" Integer "," Integer "]"
	
	ParameterList : $\varepsilon$ ; Parameter ; ParameterList "," Parameter
	
	Parameter : "payable" . VariablePartyPair ; VariablePartyPair
	
	VariablePartyPair : SingletonName  ":" "caller" ; SingletonName ":" IdVariableName 
	
	SingletonName: NumericVariableName ; IdVariableName ; MapVariableName "[" IdVariableName "]" 
	
	ConcurrentParamList : ConcurrentParam ; ConcurrentParamList "," ConcurrentParam
	
	ConcurrentParam : ConcurrentPay ; ConcurrentDecision
	
	ConcurrentPay :  "payable" . SingletonName ":" IdVariableName
	
	ConcurrentDecision : SingletonName ":" IdVariableName " = " DefaultValue
	
	CommandList : Command ; Command . CommandList
	
	Command : If ; Assignment ; Payout ; "return" ";"
	
	If : IfWithElse ; IfWithoutElse
	
	IfWithElse : "if" "(" Bexpr ")" "\{" . CommandList . "\}" "else" "\{". CommandList . "\}"
	
	IfWithoutElse : "if" "("  Bexpr ")" "\{" . CommandList . "\}"
	
	Assignment : NumericVariableName . "=" . Expr ";" ; MapVariableName"[" IdVariableName "]" "=" Expr ";" ; IdVariableName "=" IdVariableName ";" ; IdVariableName "=" "null" ";"
	
	Payout : "payout" "(" Party "," Expr ")" ";" 
	
	Party : "caller" ; IdVariableName ; "null" 
	
	Bexpr : Literal ; "(" Bexpr . "and" . Bexpr ")" ; "(" Bexpr . "or" . Bexpr ")" ; "not" . "(" Bexpr ")"
	
	Literal : Expr . Cmp . Expr ; Party "==" Party
	
	Cmp : "<" ; ">" ; "<=" ; ">=" ; "==" ; "!="
	
	Expr : NumericVariableName ; MapVariableName "[" IdVariableName "]" ; Integer ; "(" Expr . Operation . Expr ")"
	
	Operation : "+" ; "-" ; "*" ; "/"
	
\end{grammar}

Note that ``/'' is considered to be integer division. Moreover, multi-party functions are not allowed to use the keyword ``caller'' because it is ambiguous.

\subsection{Appendix to Section \ref{sec:semantics}: Formal Semantics} \label{app:semantics}

In order to define runs of programmed contracts, we first formally introduce the notions of last sets and control flow graphs.

\subsubsection{Last Sets} For an entity $\mathsf{e}$, according to the grammar above, we define the set $\mathsf{Last(e)}$ to contain those labels that can potentially form its last part in an execution. Formally, the last sets are defined by the following attribute grammar:

\begin{grammar}
	[(colon){$\rightarrow$}]
	[(semicolon){ $|$}]
	[(period){ $~$}]
	[(quote){`}{'}]
		
	Function : OnePartyFunction ; MultiPartyFunction
	
	Last(Function) = Last(OnePartyFunction) ; Last(MultiPartyFunction)
	\\ \\
	OnePartyFunction : FunctionHeader . "("  ParameterList  ")" "\{" CommandList "\}"
	
	Last(OnePartyFunction) = Last(CommandList)
	\\ \\
	MultiPartyFunction : FunctionHeader  "("  ConcurrentParamList ")" "\{" CommandList"\}"
	
	Last(MultiPartyFunction) = Last(CommandList)
	\\ \\
	CommandList$_0$ : Command ; Command . CommandList$_1$
	
	Last(CommandList$_0$) = Last(Command) ; Last(CommandList$_1$)
	\\ \\
	Command : If ; Assignment ; Payout ; "return" ";"
	
	Last(Command) = Last(If) ; $\{$ Label(Command) $\}$ ; $\{$ Label(Command) $\}$ ; $\{$ Label(Command)$\}$
	\\ \\
	If : IfWithElse ; IfWithoutElse
	
	Last(If) = Last(IfWithElse) ; Last(IfWithoutElse)
	\\ \\
	IfWithElse : "if" "(" Bexpr ")" "\{"  CommandList$_0$  "\}" "else" "\{" CommandList$_1$  "\}"
	
	Last(IfWithElse) = Last(CommandList$_0$) $\cup$ Last(CommandList$_1$)
	\\ \\
	IfWithoutElse : "if" "("  Bexpr ")" "\{" . CommandList . "\}"
	
	Last(IfWithoutElse) = Last(CommandList) $\cup$ $\{$ Label(IfWithoutElse) $\}$
\end{grammar}

\subsubsection{Control Flow Graphs} \label{app:cfg}
Given a contract and a function $f_i$, the respective control flow graph $CFG_i = (V, E)$ is a directed graph in which $V$ consists of all the labels in $f_i$ and $E$ consists of the following directed edges:
\begin{itemize}
	\item \emph{Initial edges.}
	If the first label in the CommandList of $f_i$ is $l$, then there is an edge from $\avval_i$ to $l$ with no condition, i.e.~with true condition.
	\item \emph{Return edges.} If $l$ is the label of a return statement, then there is only one edge $(l, \akhar_i)$ with no condition and none of the following cases apply.
	\item \emph{Normal flow edges.} For every label $l$ of an Assignment or a Payout, if the entity labeled $l+1$ is in the same scope, i.e.~if there are no closed braces between the two, then there is an edge $(l, l+1)$ with no condition.
	\item \emph{Conditional flow edges.} If $l$ is a label for an ``if'' statement
	\begin{itemize}
		\item If $l^+$ is the first label of $f_i$ after all parts of $l$, i.e.~ after its CommandList(s), then there is an edge from every $l'$ in the last set(s) of the CommandList(s) of $l$ to $l^+$. These edges have no conditions.
		\item If $l^t$ is the first label in first CommandList of $l$, then there is an edge $(l, l^t)$ with condition equal to the Bexpr of $l$. 
		\item Let $l^f$ be the first label in second CommandList of $l$, or if no such label exists $l^f = l^+$. Then there is an edge $(l, l^f)$ with condition equal to the negation of the Bexpr of $l$.
		\end{itemize}
		\item \emph{Final edges.} If $l \in last(f_i)$, then there is an edge $(l, \akhar_i)$ with no condition.
		\end{itemize}
\subsubsection{Formal Definition of a Run} \label{app:run}

 A run $\rho$ of the contract is a finite sequence $\left\{\rho_j = (t_j, b_j, l_j, \val_j, c_j) \right\}_{j=0}^{r}$ of states that corresponds to an execution of the contract. Formally, $\rho$ is a run if the following conditions hold:
 \begin{itemize}
 	\item \emph{Initial state.} $t_0 = b_0 = l_0 = 0$, $\val_0 = X_0$ and $c_0 = \perp$.
 	\item \emph{Final state.} $\pho_r$ must be the only state in which the time stamp $t_r$ is greater than all upper time limits $\Tup(f_i)$ of functions $f_i$ for $1 \leq i \leq n$.
 	\item For every $0 \leq j < r$, $\rho_{j+1}$ must be obtained from $\rho_j$ by one of the following rules:
 	\begin{itemize}
 		\item \emph{Ticks of the clock.} $l_j \in \{0, \akhar_1, \akhar_2, \ldots, \akhar_n\}$ and $\pho_{j+1} = (t_j + 1, b_j, 0, \val_j, \perp)$. Intuitively, this means that the clock can only tick when no function is being executed and a tick of the clock switches the contract to global scope, i.e.~outside all functions.
 		\item \emph{Execution of one-party functions.} $l_j = 0$, $c_j = \perp$, $f_i$ is a one-party function such that $\Tdown(f_i) \leq t_j \leq \Tup(f_i)$, $p$ is party and $\pho_{j+1} = (t_j, b_j, \avval_i, \val, p)$. This models the case when no function is being run and the party $p$ calls function $f_i$. In the bounded modeling case, in order to avoid the same party running the same function at the same time stamp more than once, we consider a function call as a tuple $(t_j, f_i, p)$ and require that the current function call is lexicographically later than the previous one, if it exists. Formally, if $\pho_j = (t_j, b_j, \akhar_k, \val, p_j)$ and $(t_j, f_i, p)$ is lexicographically later than $(t_j, f_k, p_j)$ then it is allowed to have $\pho_{j+1} = (t_j, b_j, \avval_i, \val, p)$. In normal execution of the programs, this ordering is not enforced.
 		\item \emph{Execution of multi-party functions.} $l_j = 0$, $c_j = \perp$, $f_i$ is a multi-party function, $t_j = \Tdown(f_i)$, and $\pho_{j+1} = (\Tup(f_i), b_j, \avval_i, \val_j, p)$ for some party $p$. This models the execution of function $f_i$ when the time $\Tup(f_i)$ is reached.
 		\item \emph{Transitions inside functions.} 
 		\begin{itemize}
 			\item \emph{First transitions in functions.} If $l_j = \avval_i$, let $\mathcal{P} \subseteq X$ be the set of variables that are designated to be paid at the beginning of $f_i$, and similarly, define $\mathcal{D}$ as the set of variables that are decided at the beginning of the function $f_i$. Then it is allowed to have $\pho_{j+1} = (t_j, b_j + \sum_{x \in \mathcal{P}} \val_{j+1}(x), l_{j+1}, \val_{j+1}, c_j)$ if the edge $(l_j, l_{j+1})$ is present in $CFG_i$ and its condition evaluates to true under $\val_j$ and for all variables $x \not\in \mathcal{P} \cup \mathcal{D}$, $\val_{j+1}(x) = \val_j(x)$ and for all $x \in \mathcal{P} ,\max\{0, \Rdown(x)\} \leq \val_{j+1}(x) \leq \Rup(x)$. The latter condition is because one cannot pay a negative amount.
 			\item \emph{Program flow transitions.} If $l_j$ is a label in $f_i$ and $l_j \neq \avval_i$,
 			\begin{itemize}
 				\item if $l_j$ is the label of an Assignment of the form $x = expr;$, then it is allowed have $\pho_{j+1} = (t_j, b_j, l_{j+1}, \val_{j+1}, c_j)$ if $(l_j, l_j+1)$ is an edge of $CFG_i$ and for all variables $y \neq x$, $\val_{j+1}(y) = \val_j(y)$ and $\val_{j+1}(x) = \max\{ \Rdown(x), \min\{ \Rup(x) ,  \val_j(expr) \} \}$. This ensures that all variables remain in bounds. Intuitively, overflows and underflows cause the variable to store its highest possible, resp. lowest possible, value. Similarly, if $l_j$ is the label of an assignment to an id variable, then the valuation must remain the same, except that the value of that id variable must be updated accordingly.
 				
 				\item if $l_j$ is the label of an If statement, either with or without an else part, then $\pho_{j+1} = (t_j, b_j, l_{j+1}, \val_j, c_j)$ is allowed only if $(l_j, l_{j+1})$ is an edge of $CFG_i$ and its condition evaluates to true under $\val_j$ and $c_j$.
 				
 				\item if $l_j$ is the label of a Payout statement of the form payout $x$ to $p$ for some expression $x$ and party $p$, defining $\omega = \min \{b_j, \max\{0, \val_j(x)\}\}$, it is allowed to have $\pho_{j+1} = (t_j, b_j - \omega, l_{j+1}, \val_j, c_j)$ if $(l_j, l_{j+1})$ is an edge of $CFG_i$. This makes sure that payouts are always nonnegative and do not exceed the current balance of the contract.
 				
 				\item if $l_j$ is the label of a return statement and has an edge to $l'$ in the control flow graph of the current function, then we can have $\pho_{j+1} = (t_j, b_j, l', \val, p)$.
 			\end{itemize}
 		\end{itemize}
 	\end{itemize}
 \end{itemize}

\section{Appendix to Section \ref{sec:games}: Bounded Analysis and Games}
\subsection{Appendix to Section \ref{sec:concur_games}: Computing Game Values} \label{app:games}

\smallskip\noindent{\em Notation.} We use $u_\fhlen(\sigma_1, \sigma_2)$ to denote $\E^{(\sigma_1, \sigma_2)} \left[u_\fhlen(p)\right]$.

\smallskip\noindent \emph{Best Responses.} We say that $\sigma_2$ is a best response to $\sigma_1$ if $u_\fhlen(\sigma_1, \sigma_2) = \inf_{\varsigma_2} u_\fhlen(\sigma_1, \varsigma_2)$. It is easy to see that for each $\sigma_1$ there exists a pure best response, i.e.~a best response that does not use randomization, because given $\sigma_1$, the utility of every $\sigma_2$ is an affine combination of utilities of some pure strategies for player~2.

\smallskip\noindent\emph{Stateless Games.} A stateless (matrix) game between two players is a tuple $\matgame = (A_1, A_2, v)$ where $A_i$ is the set of actions available to player~$i$, and $v: A_1 \times A_2 \rightarrow \mathbb{R}$ is a utility function. In a stateless game, a strategy for player $i$ is simply a member of $A_i$ and a mixed strategy $\varsigma_i$ for player $i$ is a probability distribution over $A_i$. As usual, the game is zero-sum, i.e., player~1 is trying to maximize the utility and player~2 to minimize it. The value of a stateless game $\matgame$ is defined as $ \gamevalue(\matgame) := \sup_{\varsigma_1} \inf_{\varsigma_2} u(\varsigma_1, \varsigma_2) = \inf_{\varsigma_2} \sup_{\varsigma_1} u(\varsigma_1, \varsigma_2)$. It is well-known that the latter equality , which is called determinacy theorem, holds and that game values for stateless games can be computed efficiently by linear programming.

Let $G = (S, s_0, A, \Gamma_1, \Gamma_2, \delta)$ and $s \in S$, then by $G_s$ we mean the game structure $G$ starting at state $s$, i.e.~$G_s = (S, s, A, \Gamma_1, \Gamma_2, \delta)$. Of course $G = G_{s_0}$.

\smallskip\noindent\emph{Local Games.} Let $(G, u)$ be a concurrent game and $s \in S$, then we denote by $G[s, t] = (A_1, A_2, v)$, the local stateless game at state $s$, when there are $t$ more steps left. We now formalize this. It must be the case that $A_1 = \Gamma_1(s), A_2 = \Gamma_2(s)$ and for every pair of actions $(a_1, a_2)$ in $A_1 \times A_2$, if $s'$ is the successor state of $s$ with these actions, i.e.~if $\delta(s, a_1, a_2) = s'$, then $v(a_1, a_2) = u(s) + \gamevalue_{t-1}(G_{s'}, u)$. It is easy to check that by definition, $\gamevalue(G[s, t]) = \gamevalue_t(G_s, u)$.

\smallskip\noindent\emph{Value Iteration.} The last equality above leads to a natural algorithm for computing game values. In order to compute $\gamevalue_\fhlen(G, u)$, we compute $\gamevalue_t(G_s, u) = \gamevalue(G[s, t])$ for every state $s \in S$ and $t \leq \fhlen$. This algorithm is called value iteration \cite{FV97} and is illustrated as Algorithm \ref{algo:vi}. 

\begin{algorithm}
	\caption{Value Iteration}\label{algo:vi}
	\begin{algorithmic}[1]
		\Procedure{ValueIteration}{$G = (S, s_0, A, \Gamma_1, \Gamma_2, \delta)$, $u$, $\fhlen$}
		\For{$t \in \{0, 1, \ldots, \fhlen\}$}
		\For {$s \in S$}
		\If {$t=0$}
		\State $\gamevalue_t(G_s, u) \gets 0$
		\Else
		\State Create local game $G[s, t]$
		\State Compute $\gamevalue(G[s, t])$ by linear programming
		\State $\gamevalue_t(G_s, u) \gets \gamevalue(G[s, t])$
		\EndIf
		\EndFor
		
		\State \Return $\gamevalue_\fhlen(G_{s_0})$
		\EndFor
		\EndProcedure
	\end{algorithmic}
\end{algorithm}

\smallskip\noindent\emph{Corollary.} The determinacy theorem for stateless local games and value iteration lead to the following determinacy theorem: For every concurrent game $(G, u)$,
$$
\gamevalue_\fhlen(G, u) = \sup_{\sigma_1} \inf_{\sigma_2} u_\fhlen(\sigma_1, \sigma_2) = \inf_{\sigma_2} \sup_{\sigma_1} u_\fhlen(\sigma_1, \sigma_2).
$$

\subsection{Appendix to Section \ref{SEC:TRANS}: Translation of Contracts to Games} \label{app:contogame}

In this section we formally show how a bounded contract $C_k$ can be converted to a game and establish relationships and correspondences between the two.

Let $C = (N, I, M, R, X_0, F, T)$ be a contract, $C_k$ its bounded model with $k$ parties, and $S$ the set of all states that can appear in runs of $C_k$. Also, let $\party \in \parties$ be a party with an objective $o$. 
We construct a game $(G = (S \cup \{\dummy \}, s_0, A, \Gamma_1, \Gamma_2, \delta), u)$ corresponding to $(C_k, \party, o)$ as follows:
\begin{itemize}
	\item Each state of the game corresponds to a state of the contract, except for $\dummy$ which is a dummy state used to control the end of the contract. 
	\item $s_0$ is the initial state of $C_k$, i.e.~$s_0 = (0, 0, 0, X_0, \perp)$. This is the state from which all runs of $C_k$ begin.
	\item Player~1 of the game corresponds to the party $\party$ and player~2 corresponds to all other parties colluding together.
	\item Each action of the game corresponds to either a single move in the contract or a tuple of $k-1$ moves, i.e.~$A = \moves \cup \moves^{k-1}$. Player~1 actions are only single-move, to model moves by $\party$, and player~2 actions are tuples of $k-1$ moves, to model the joint action of the $k-1$ other parties.
	\item At each state $s \in S$, $\Gamma_1(s) = P_\party(s)$. Intuitively, this means that the actions available to player~1 at state $s$ of the game correspond to permitted moves of $\party$ in state $s$ of the contract.
	\item Similarly, at each state $s \in S$, $\Gamma_2(s) = \prod_{\party' \in \parties \setminus \{\party\}} P_{\party'}(s)$, i.e., each action of player~2 at state $s$ of the game corresponds to a joint action of all parties other than $\party$ in state $s$ of the contract.
	\item $\Gamma_1 (\dummy) = \Gamma_2(\dummy) = \{ \naught \}$.
	\item $\delta$ follows the rules of the contract. More precisely, if in an execution of the contract at state $s$, party $\party$ chooses the permitted move $m_\party \in P_\party(s)$ and other parties, $\party'_1, \party'_2, \ldots, \party'_{k-1}$ choose permitted moves $m'_1, m'_2, \ldots, m'_{k-1}$, respectively, and the execution continues with the new state $s'$, then in the game, we define $\delta(s, m_\party, (m'_1, m'_2, \ldots, m'_{k-1})) = s'$. Moreover, if $s$ is the last state of a run of $C_k$, then $\delta(s, m_\party, (m'_1, m'_2, \ldots, m'_{k-1})) = \dummy$ and $\delta(\dummy, \naught, \naught) = \dummy$.
	\item $u$ mimics $o$, i.e.~for every state $s$:
	\begin{itemize}
		\item If $s$ is the last state of a run of $C_k$ and $o$'s expression part, i.e.~the part other than $(\party^+ - \party^-)$, evaluates to $x$ under the valuation at $s$, then in the game we have $u(s) = x$.
		\item If $o$ has a $(p^+ - p^-)$ part and the label in $s$ points to a payout statement of the contract that pays a value of $x$ under $s$'s valuation to $\party$, then $u(s) = x$.
		\item If $o$ has a $(p^+ - p^-)$ part and $s$ is a successor of a payment/decision in the contract in which $\party$ has paid a total value of $x$ under the corresponding valuation, then $u(s) = -x$.
		\item If $s$ satisfies several of the conditions above, then $u(s)$ is the sum of values defined in all those cases. If it satisfies none of them, $u(s) = 0$.
	\end{itemize}
\end{itemize}

By simply following the definitions, it can be settled that (i)~every policy profile $\pi$ of $C_k$ corresponds to a strategy profile $\purestrategy$ in $G$, in which player~1 acts in correspondence with the moves of $\party$ and player~2 acts in correspondence with the joint moves of other parties, (ii)~ every randomized policy $\xi_\party$ in the contract corresponds naturally to a mixed strategy for player~1 in the game and vice-versa, (iii) every randomized policy profile for parties other than $\party$ in the contract corresponds to a mixed strategy for player~2 in the game, in which actions are distributed according to the distributions for other parties' moves. However, this correspondence is not surjective, given that different components in a mixed strategy for player~2 need not be independently distributed, and (iv) The outcome of a policy profile, resp. randomized policy profile, in $C_k$ is equal to the utility of its corresponding strategy profile, resp. mixed strategy profile, in the game.

Note that these results hold only if our finite horizon $\fhlen$ is big enough to model all the runs as plays of the game. This can be achieved by letting $\fhlen$ be at least as big as the length of the longest run, which according to our semantics is no more than $\vert \parties \vert \times t \times l$, where $\parties$ is the set of parties, $l$ is the number of labels and $t$ is the largest time in the contract. If $\fhlen$ is greater than the length of a run, then that run corresponds to a play of the game that reaches the dummy state $\dummy$ and stays there. 

\begin{theorem}[Correspondence] 
	If $(G, u)$ is the game corresponding to the bounded contract $C_k$ for a party $\party$ with objective $o$, then the value of the game, $\gamevalue(G, u)$ is equal to the value of the bounded contract, $\mathsf{V}(C_k, o, \party)$.
\end{theorem}
\begin{proof} Note that by definition 
	$$\mathsf{V}(C_k, o, \party) = \sup_{\xi_\party \in \Xi_\party} \inf_{\xi_{-\party} \in \Xi_{-\party}} \kappa(\xi_\party, \xi_{-\party}, o, \party),$$
	where by $\kappa(\xi_\party, \xi_{-\party}, o, \party)$ we mean the expected outcome when using $(\xi_\party, \xi_{-\party})$.
	Given that the set of all randomized policy profiles is compact and $\kappa$ is a bounded continuous function, there must exist a specific randomized policy profile $\upbeta = (\upbeta_\party, \upbeta_{-\party})$
	such that $\mathsf{V}(C_k, o, \party) = \kappa(\upbeta, o, \party)$. According to (iv), there exists a mixed strategy profile $\upvarsigma = (\upvarsigma_1, \upvarsigma_2)$ in the game, such that $\kappa(\upbeta, o, \party) = u(\upvarsigma)$. We prove that $\upvarsigma_2$ is a best response to $\upvarsigma_1$ in $(G, u)$. Let $\upvarsigma_2'$ be a pure best response to $\upvarsigma_1$, then according to (i) and (ii), there exists a randomized policy profile $\upbeta' = (\upbeta_\party, \upbeta'_{-\party})$ such that $\kappa(\upbeta', o, \party) = u(\upvarsigma_1, \upvarsigma'_2)$, but by definition of $\upbeta$, $\kappa(\upbeta', o, \party) \geq \kappa(\upbeta, o, \party)$ so $u(\upvarsigma_1, \upvarsigma_2') \geq u(\upvarsigma) $ and hence $\upvarsigma_2$ is a best response to $\upvarsigma_1$. So $\mathsf{V}(C_k, o, \party) = \kappa(\upbeta, o, \party) = \inf_{\sigma_2} u(\upvarsigma_1, \sigma_2) \leq \sup_{\sigma_1} \inf_{\sigma_2} u(\sigma_1, \sigma_2) = \gamevalue(G, u)$.
	
	We now prove the other side. A similar argument as in the previous case shows that there must exist a mixed strategy profile $\upvarsigma = (\upvarsigma_1, \upvarsigma_2)$ such that $u(\upvarsigma) = \sup_{\sigma_1} \inf_{\sigma_2} u(\sigma_1, \sigma_2) = \gamevalue(G, u)$, given that changing $\upvarsigma_2$ with any other best response for $\upvarsigma_1$ does not change $u(\upvarsigma)$, wlog we can assume that $\upvarsigma_2$ is a pure best response. Then, according to (1) and (2) above, there exists a randomized policy profile $\upbeta = (\upbeta_\party, \upbeta_{-\party})$ such that $ u(\upvarsigma) = \kappa(\upbeta_\party, \upbeta_{-\party}, o, \party)$. We claim that $\kappa(\upbeta_\party, \upbeta_{-\party}, o, \party) = \inf_{\xi_{-\party}} \kappa(\upbeta_\party, \xi_{-\party}, o, \party)$, because otherwise one can create a better response for $\upvarsigma_1$ than $\upvarsigma_2$ using the process explained in the previous part. Therefore, $\gamevalue(G, u) = u(\upvarsigma) = \inf_{\xi_{-\party}} \kappa(\upbeta_\party, \xi_{-\party}, o, \party) \leq \sup_{\xi_\party} \inf_{\xi_{-\party}} \kappa(\xi_\party, \xi_{-\party}, o, \party) = \mathsf{V}(C_k, o, \party).$
\end{proof}

\section{Appendix to Section \ref{sec:exp}: Details of Implementation and Experimental Results} \label{app:backward_induction} \label{app:exper}

\subsection{Appendix to Section \ref{sec:impl}: Implementation Optimizations}

Besides the optimization of first applying the abstraction, and then 
obtaining the games, we also apply other optimizations to solve
the abstracted games faster. 
They are as follows:
\begin{itemize}
	\item {\em Dummy state removal.} First as explained in Remark~\ref{rem:dummy}
	the dummy states of the abstraction are for conceptual clarity, and 
	we do not construct them explicitly as part of state space.
	
	\item {\em Exploiting acyclic structure.} 
	While for general concurrent games with cycles the standard procedure 
	to obtain values is by value-iteration algorithm, we exploit the 
	structure of games obtained from contracts.
	Since time is an important aspect of the semantics of contracts, 
	and is encoded as part of the state space, the concurrent games obtained 
	from contracts and their interval-abstraction are acyclic. 
	For acyclic games, instead of general value-iteration, backward induction 
	approach can be applied.
	The backward induction approach works for acyclic games, and it computes
	values bottom-up (according to the topological sorting order), i.e.,
	it computes value of a state of the games only when the values of all 
	successors are determined.
	Details of the backward-induction is presented below.
\end{itemize}
The above optimizations ensure that we can solve the abstract games 
quite efficiently to obtain their values.

\subsection{Acyclic Game Structures and Backward Induction}
In this section we provide an overview of acyclic game structures and the classic backward induction algorithm. 

\smallskip\noindent\emph{Acyclic Game Structures.} An acyclic concurrent game structure is a tuple $G = (S, s_0, A, \Gamma_1, \Gamma_2, \delta)$ where all the parts have the same meaning and definition as in concurrent games, except that (i)~$\Gamma_i(s)$ is now allowed to be empty, i.e., it is possible to have dead-end states in which at least one of the players has no available actions and (ii)~no state is reachable from itself using a non-empty sequence of transitions of the game.

\smallskip\noindent\emph{Plays.} A play of an acyclic game structure $G$ is a finite sequence of states and valid action pairs starting from $s_0$, where each state is the result of transitioning from its previous state and action pair and the sequence ends in a dead-end.

\smallskip\noindent\emph{Strategies and Mixed Strategies.} A strategy $\purestrategy_i$ for player~$i$ in an acyclic game structure $G$ is a function $\purestrategy_i : \Histories \rightarrow A$ suggesting a valid action to player~$i$ after each history of the game. However, given that in an acyclic concurrent game the future unfolding of a play is not dependent on how the current state was reached, we can restrict our focus to the set of memoryless strategies, i.e., strategies that only depend on the current state of the game and not the whole history. A memoryless strategy is a function of the form $\purestrategy_i : S \rightarrow A$. In the same spirit, a memoryless mixed strategy for player~$i$ is a function $\sigma_i: S \rightarrow \Delta(A)$, that assigns a probability distribution over permitted actions of player~$i$ at each state. A mixed strategy profile $\sigma = (\sigma_1, \sigma_2)$ induces a unique probability distribution $\Prob^\sigma \left[\cdot\right]$ over the set of all runs. In the remainder of this section, all strategies are considered to be memoryless.

\smallskip\noindent\emph{Utilities.} A utility function $u$ for player~$1$ assigns a real utility $u(s)$ to each state $s$ of the game. Utility of a play $\pi$ is defined as the sum of utilities of its states.

\smallskip\noindent\emph{Game Value.} Given an acyclic game structure $G$, together with a utility function $u$, the value of the game $(G, u)$ is defined as:
$$
\gamevalue(G, u) = \sup_{\sigma_1} \inf_{\sigma_2} \E^{(\sigma_1, \sigma_2)} \left[ u(\pi) \right]
$$  
where $\sigma_i$ iterates over the set of all mixed strategies of player~$i$.

\noindent\emph{Local Games.} Let $(G, u)$ be an acyclic concurrent game and $s \in S$, then we denote by $G[s] = (A_1, A_2, v)$, the local stateless game at state $s$. It must be the case that $A_1 = \Gamma_1(s), A_2 = \Gamma_2(s)$ and for every pair of actions $(a_1, a_2)$ in $A_1 \times A_2$, if $s'$ is the successor state of $s$ with these actions, i.e.~if $\delta(s, a_1, a_2) = s'$, then $v(a_1, a_2) = u(s) + \gamevalue(G_{s'}, u)$. It is easy to check that by definition, $\gamevalue(G[s]) = \gamevalue(G_s, u)$.

\noindent\emph{Backward Induction.} The last equality above leads to a natural algorithm for computing game values. In order to compute $\gamevalue(G, u)$, we compute $\gamevalue(G_s, u) = \gamevalue(G[s])$ for every state $s \in S$. This algorithm is called backward induction and is illustrated as Algorithm \ref{algo:induction}. 

\begin{algorithm}
	\caption{Backward Induction}\label{algo:induction}
	\begin{algorithmic}[1]
		\Procedure{BackwardInduction}{$G = (S, s_0, A, \Gamma_1, \Gamma_2, \delta)$, $u$}
		
		\State Sort $S$ in Topological Order
		\For {$s \in S$}
		\If {$\Gamma_1(s) = \emptyset$ \textbf{or} $\Gamma_2(s) = \emptyset$}
		\State $\gamevalue(G_s, u) \gets u[s]$
		\Else
		\State Create local game $G[s]$
		\State Compute $\gamevalue(G[s])$ by linear programming
		\State $\gamevalue(G_s, u) \gets \gamevalue(G[s])$
		\EndIf
		\EndFor
		
		\State \Return $\gamevalue(G_{s_0})$
		\EndProcedure
	\end{algorithmic}
\end{algorithm}

\subsection{Appendix to Section \ref{sec:res}: Modelling Details in Experiments} \label{app:model}

\smallskip\noindent{\em Rock-Paper-Scissors.} For this contract, we consider the correct implementation as in Figure \ref{prog:rps} along with a classic buggy variant that allows sequential, instead of simultaneous, interactions with the contract. We analyze the case where $k = 2$, because rock-paper-scissors is essentially a two-player game and adding more parties has no effect other than increasing the complexity. The contract is analyzed from the point-of-view of its issuer, Alice. Her utility is defined to be her overall monetary payoff plus $10$ units if she wins the game. 

\smallskip\noindent\emph{Auction.} We consider the Auction contract and its buggy version as in Figure \ref{prog:auction}. The analysis models utility of a party as her monetary gain (or loss) plus value of the auctioned object if she wins the auction. In the buggy version, one can reduce her bid freely. In this case analyzing the contract with only one party, i.e., with $k = 1$, already detects the bug and leads to a value strictly greater than $0$, while a correct contract should always have an objective value of $0$.  

\smallskip\noindent\emph{Lottery.} We consider a lottery contract, and its wrong implementation, as in Figure \ref{prog:lottery} and analyze it from issuer's viewpoint assuming that the only utility is the monetary gain or loss that occurs during the contract. This example is different from other examples because variables have quite small ranges and even a single refinement step can find precise values. Hence, as shown in the results, there is no need of further refinement.

\smallskip\noindent\emph{Token sale.} In the token sale example of Figure \ref{prog:sale}, we model utility as the total number of sold tokens. In a bug-free system this should never become more than the number of available tokens, which in this case is fixed to $1000$. However, a buggy implementation might lead to a higher payoff.

\smallskip\noindent\emph{Token transfer.} A similar analysis with the same objective function as in token sale for a contract which supports both sale and transfer of tokens (Figure \ref{prog:tran}), can detect a bug that creates extra tokens in some transfers. The bug is again due to the fact that values are strictly greater than $1000$.

\end{document}